\documentclass[12pt, reqno]{amsart}

\makeatletter
\g@addto@macro{\endabstract}{\@setabstract}
\makeatother

\usepackage{graphicx}
\usepackage{amsmath, amssymb, amsthm}
\usepackage{natbib}

\usepackage[usenames,dvipsnames,svgnames,table]{xcolor}

\usepackage[mathscr]{euscript} 
\usepackage{bbm}

\usepackage{enumitem}
\setlist[enumerate]{itemsep=5pt,topsep=3pt}
\setlist[itemize]{itemsep=2pt,topsep=3pt}
\setlist[enumerate,1]{label=\arabic*.}

\renewcommand{\leq}{\leqslant}
\renewcommand{\geq}{\geqslant}

\renewcommand{\phi}{\varphi}
\renewcommand{\epsilon}{\varepsilon}

\usepackage{mdwlist}

\usepackage[citecolor=blue, colorlinks=true, linkcolor=blue]{hyperref}

\usepackage{caption}
\usepackage{subcaption}
\usepackage{algorithmic}
\usepackage{algorithm}

\usepackage[left=1.3in, right=1.3in, top=1in, bottom=0.9in, includehead, includefoot]{geometry}

\DeclareMathOperator*{\argmin}{arg\,min}

\DeclareMathOperator{\interior}{int}

\newcommand{\setntn}[2]{ \{ #1 : #2 \} }

\newcommand{\iI}{\mathscr I}

\newcommand{\mM}{\mathscr M}

\newcommand{\fF}{\mathscr F}

\newcommand{\RR}{\mathbbm R}
\newcommand{\NN}{\mathbbm N}

\newcommand{\ZZ}{\mathbbm Z}

\theoremstyle{plain}
\newtheorem{theorem}{Theorem}[section]

\newtheorem{lemma}[theorem]{Lemma}
\newtheorem{proposition}[theorem]{Proposition}

\theoremstyle{definition}
\newtheorem{definition}{Definition}[section]

\newtheorem{assumption}{Assumption}[section]

\begin{document}

\title{}

\date{\today}

\begin{center}
    \Large Coase Meets Bellman: Dynamic Programming for \\ Production
    Networks\footnote{We gratefully acknowledge valuable comments from our
        editors and referees, as well as Pol
      Antr{\`a}s, Yongsung Chang, Davin Chor, Ryo Jinnai,
        Young Sik Kim, Yuta Takahashi and Makoto Saito.  We also acknowledge financial
      support from Korea University (K1922081), BrainKorea21 Plus (K1327408 and T192201), JSPS Grant-in-Aid for
      Scientific Research (20H05633 and 16H03598), Singapore Ministry of
      Education Academic Research Fund R-122-000-140-112 and Australian
      Research Council Award DP120100321. Email contacts:
      \texttt{tomookikuchi@waseda.jp},
      \texttt{nishimura@rieb.kobe-u.ac.jp},
      \texttt{john.stachurski@anu.edu.au}, \texttt{zhangjunnan1224@gmail.com}
    }

  \large
    \bigskip
  Tomoo Kikuchi\textsuperscript{a}, Kazuo Nishimura\textsuperscript{b}, 
  \\
  John Stachurski\textsuperscript{c} and Junnan Zhang\textsuperscript{d}
  {\let\thefootnote\relax\footnotetext{
  \textsuperscript{a} Graduate School of Asia-Pacific Studies, Waseda University;
  \textsuperscript{b} Research Institute of Economics and
      Business Administration, Kobe University, and RIETI; 
  \textsuperscript{c} Research School of Economics, Australian
      National University;
  \textsuperscript{d} Center for Macroeconomic Research at School of
  Economics, and Wang Yanan Institute for Studies in Economics, Xiamen
  University}
  }
  \par \bigskip

  \normalsize

  \today
\end{center}

\begin{abstract} 
    We show that competitive equilibria in
    a range of models related to production networks
     can be recovered as solutions to dynamic programs.
    Although these programs fail to be contractive, we prove that they are
    tractable.  As an illustration, we treat Coase's theory of the
    firm, equilibria in production chains with transaction
    costs, and equilibria in production
    networks with multiple partners.  We then show how the same techniques 
    extend to other equilibrium and decision problems, such as 
    the distribution of management layers within firms and the spatial
    distribution of cities.

    \vspace{1em}

    \\
    \noindent
    {\bf JEL\ Classification:} C61, D21, D90
\end{abstract}

%The results obtained in the paper provide several benefits.  One is that the
%class of dynamic programming problems we consider is of interest in its own
%right, since it provides a foundation for the theory of choice in a commonly
%observed dynamic setting (negative discounting) where traditional optimization
%methods are problematic---see below for more discussion.  A second more
%obvious benefit is that dynamic programming theory can be brought to bear on
%the theory of the firm and the structure of production chains, supplying both
%analytical and computational methods.  

\section{Introduction}

Production networks have grown rapidly in size and complexity, in line with
advances in communications, supply chain management and transportation
technology (see, e.g.,~\cite{coe2015global}). These large and complex
networks are sensitive to uncertainty, trade disputes, transaction
costs and other frictions.  Firms routinely shift production and task
allocation across networks, in order to mitigate risk or exploit new
opportunities (see, e.g.,~\cite{bmk}). There is an ongoing need to predict how
equilibria in production networks adapt and respond to shocks, in order to
understand their impact on domestic employment, industry concentration,
productivity and tax revenue.

Dynamic programming provides one methodology for analyzing such equilibria.
While dynamic programming is typically used to study
\emph{dynamic} models (see, e.g., \cite{stokeyec}), it
can also be applied to static models by reinterpreting the time
parameter as an index over firms or other decision making entities, as seen
in, for example, \cite{garicano2006organization}, \cite{hsu2014optimal},
\cite{tyazhelnikov2019production}, and \cite{antras2020geography}.  Our paper
builds on this literature by providing a systematic way to apply the theory of
dynamic programming to both production chains and production networks, as well
as to a range of other static allocation problems involving firm management
and economic geography.

This research agenda faces a technical hurdle: the dynamic programs
most naturally mapped to the competitive allocation problems we wish to
consider usually fail to be contractive.  Contractivity fails because
frictions such as the transaction costs or failure probabilities in the
production chain models translate into \emph{negative} discount rates in the
corresponding dynamic program.  In this paper, we circumvent the need for
contractivity by drawing on dynamic programming methods originally developed to 
solve recursive preference problems.\footnote{See, for example, 
    \cite{epstein1989substitution}, \cite{bloise2018convex} or
    \cite{marinacci2019unique}.  In this sense, our work can be viewed as
    building connections between (a) the existing literature on dynamic
    programming for obtaining static competitive equilibria and (b) the modern
theory of dynamic programming with recursive preferences.}

The contributions of this paper fall into two parts.  The
first is providing a theory of dynamic programming in a
loss-minimization setting where discount rates are negative. The
second is applying this theory to a series of competitive equilibrium problems
involving production chains, production networks and other related models.
Through the application of this theory, we show how the dynamic programming
tools can be used to obtain not only existence and uniqueness of equilibria,
but also computational algorithms, results on comparative statics and insights into the
underlying mechanisms.  

Regarding application, we build on an analytical framework for analyzing
allocation of tasks across firms first developed by \cite{coase1937nature}.
Subsequently, \cite{kikuchi2018span}, \cite{fally2018coasian} and
\cite{yu2019equilibrium} developed Coasian models in which firms trade off
coordination costs within the firm against transaction costs outside the firm.
We show that competitive equilibria in these models can be recovered as solutions
to dynamic programs and use the associated envelope condition to provide
insight on some of the foundational conjectures of \cite{coase1937nature}.

In the remainder of the paper, we then apply similar methods to study a
range of additional applications, including settings where Coasian
transaction costs are replaced by failures in production or costly
transportation, as found, for example, in \cite{levine2012production} and
\cite{costinot2013elementary}; models of knowledge organization and optimal
management structures originally due to \cite{garicano2000hierarchies}; the
analysis of central place theory in \cite{hsu2014optimal}; and the
configuration of general (nonsequential) production networks in the spirit
of \cite{baldwin2013spiders}, \cite{kikuchi2018span},
\cite{yu2019equilibrium} and \cite{tyazhelnikov2019production}.

The applications discussed above differ in many ways. There are different
trade-offs that characterize each model, each of which leads to a particular
endogenous structure.  The negative discount dynamic programming theory developed here provides
a unifying methodology and brings tools to bear on understanding the structure
of the networks where firms, cities and managers coordinate production.

Regarding our technical contribution, the closest existing work in the
economic literature is \cite{bloise2018convex}, who treat noncontractive
dynamic programming problems that arise from recursive utility.  In addition
to results on existence and uniqueness of fixed points of the Bellman
operator, which parallel analogous results in \cite{bloise2018convex}, we
apply a fixed point result of \cite{du1989fixed} to provide new results on
monotonicity, convexity and differentiability of solutions, as well as a full
set of optimality results linking Bellman's equation to existence and
characterization of optimal solutions.\footnote{This optimality theory is
    related to other studies of dynamic programming where the Bellman operator
    fails to be a contraction, such as \cite{martins2010existence} and
    \cite{rincon2003existence}.  Our methods differ because even the
    relatively weak local contraction conditions imposed in that line of
    research fail in our settings.  The fixed point results in this paper are
    related to those found in \cite{kamihigashi2015application}, but here we
    also prove uniqueness of the fixed point, as well as connections to
optimality and shape and differentiability properties.}

The remainder of this paper is structured as follows. In
Section~\ref{s:model}, we study a dynamic optimization problem under
negative discounting and discuss its solution. In
Section~\ref{s:app:production}, we connect this discussion to Coase's theory
of the firm and elaborate on the relationship between our model and other
related models. In Section~\ref{s:app:knowledge} we show that our model can
also be used to understand organization of knowledge within a firm. In
Section~\ref{s:app:networks} we extend our model to expand the scope of
applications to more complex networks. Section~\ref{s:c} concludes. Most
proofs are deferred to the appendix.

\section{Negative Discount Dynamic Programming}\label{s:model}

%\textcolor{blue}{
%There is one more economic application of the dynamic programming theory
%contained in the paper, which we alluded to above.  This application is
%presented as a running example in the section on dynamic programming theory
%(Section~\ref{s:model}), and concerns an agent who seeks to minimize the
%present value of a sequence of losses over an infinite horizon, while
%assigning future losses \emph{greater} weight than current losses (see, e.g.,
%Section~).  In other words, the subjective discount rate is negative.  Unlike
%the other applications in this paper, the negative discount problem is dynamic
%and concerns the actions of an individual agent.  The reason a connection
%exists between such negative discount dynamic programs and the competitive
%equilibrium problems we consider in the rest of the paper, is that, in these
%equilibrium problems, frictions such as transaction costs, transportation
%costs and positive failure rates show up in the dynamic program as a negative
%discount rate.}

In this section, we study a dynamic optimization problem in which an agent
minimizes a flow of losses under negative discounting.
  While our main aim is to develop 
techniques for calculating equilibria in production networks, the
topic of negative discount loss minimization does have some independent value.\footnote{
For example, \cite{thaler1981some}, 
\cite{loewenstein1991negative} and \cite{loewenstein1991workers} document
separate instances of such phenomena.  
    \cite{loewenstein1991workers} found that the majority of surveyed workers
    reported a preference for increasing wage profiles over decreasing ones,
    even when it was pointed out that the latter could be used to construct a
    dominating consumption sequence.  \cite{loewenstein1991negative} obtained
    similar results, stating that ``sequences of outcomes that
decline in value are greatly disliked, indicating a negative rate of time
preference'' \cite[p.~351]{loewenstein1991negative}.}

Consider an agent who takes action $a_t$ in
period $t$ with loss $\ell(a_t)$. We interpret $a_t$ as effort
and $\ell(a_t)$ as disutility. Her optimization problem is, for
some $\hat x > 0$,
\begin{equation}\label{eq:dpm}
    \min_{\{a_t\}} \;  \sum_{t=0}^{\infty} \beta^t \ell(a_t) 
    \;\; \text{ s.t. } a_t \geq 0 \text{ for all } t \geq 0
    \text{ and } \sum_{t=0}^{\infty} a_t = \hat x.
\end{equation}
Throughout this section, we suppose that 
\begin{equation}\label{eq:ndas}
    \beta > 1, \;
    \ell(0) = 0, \;
    \ell' > 0 \; 
    \text{ and }
    \ell'' > 0.   
\end{equation}
The convexity in $\ell$
encourages the agent to defer some effort. Negative discounting ($\beta > 1$)
has the opposite effect. We call problem~\eqref{eq:dpm} under the
assumptions in~\eqref{eq:ndas} a \emph{negative discount dynamic
    program}.\footnote{The assumption $\ell(0) = 0$ cannot be weakened, since
    $\ell(0) > 0$ implies that the objective function is infinite.
    Conversely, with the assumption $\ell(0)=0$, minimal loss is always
    finite.  Indeed, by choosing the feasible action path $a_0 = \hat x$ and
    $a_t = 0$ for all $t \geq 1$, we get $\sum_{t=0}^{\infty} \beta^t
    \ell(a_t) \leq \ell(\hat x)$.  Also, given our other assumptions, there is
    no need to consider the case $\beta \leq 1$ because no solution exists.
    Because we are minimizing disutility, when $\beta < 1$ any proposed solution
    $\{a_t\}$ can be strictly improved by shifting it one step into the future
    (set $a_0' = 0$ and $a_{t+1}' = a_t$ for all $t \geq 0$).  Furthermore, if
    $\beta = 1$, and a solution $\{a_t\}$ exists, then the increments
    $\{a_t\}$ must converge to zero, and hence there exists a pair $a_T$ and
    $a_{T+1}$ with $a_T > a_{T+1}$.  Since $\ell$ is strictly convex, the
objective $\sum_t \ell(a_t)$ can be reduced by redistributing a small amount
$\epsilon$ from $a_T$ to $a_{T+1}$.  This contradicts optimality.}

\subsection{A Recursive View}

We can express the problem recursively by introducing a state process
$\{x_t\}$ that starts at $\hat x$ and tracks the amount of tasks remaining.
Set $x_{t+1} = x_t - a_t$ and $x_0 = \hat x$. The Bellman equation for this
problem is
\begin{equation}
    \label{eq:epe}
    w(x) = \inf_{0 \leq a \leq x} \, \{ \ell(a) + \beta w(x-a) \}.
\end{equation}
The Bellman operator is
\begin{equation}
    \label{eq:bellop}
    (Tw)(x) = \inf_{0 \leq a \leq x} \, \{ \ell(a) + \beta w(x-a) \}.
\end{equation}
The Bellman operator is not a supremum norm
contraction because $\beta > 1$.\footnote{For example, let $w
  \equiv 1$ and $g \equiv 0$. Then $Tw \equiv \beta > 1$ while $Tg \equiv
  0$. One consequence is that, if we take an arbitrary continuous bounded
  function and iterate with $T$, the sequence typically diverges. For
  example, if $w \equiv 1$, then, $T^n w \equiv \beta^n$, which diverges to
  $+ \infty$.} Nevertheless, we can show that $T$ is well
behaved, with a unique fixed point, after we restrict its domain to a suitable
candidate class $\iI$. To this end, we set
\begin{equation*}
    X := [0, \hat{x}],
    \quad \phi(x) := \ell'(0)x 
    \quad \text{ and }  \quad
    \psi(x) := \ell(x).    
\end{equation*}
Let $\iI$ be all continuous $w \colon X \to \RR$ with $\phi \leq w \leq \psi$.
These upper and lower bounds have natural 
interpretations.  Since completing all 
remaining tasks at once is in the choice set, its value $\ell(x)$ is an upper
bound of the minimized value.  Regarding the lower bound 
$\ell'(0)x $, this is the value that could be obtained if $\beta=1$ (no
discounting) and the agent, having no time constraint, subdivided without limit.

\begin{proposition}\label{p:ndp}
    The Bellman operator has a unique fixed point $w^*$
    in $\iI$ and $T^nw \to w^*$ as $n \to \infty$ for all $w \in \iI$.
    Moreover, 
    \begin{enumerate}
        \item $w^*$ is strictly increasing, strictly convex, and
            continuously differentiable, and
        \item The policy $\pi^*(x) := \argmin_{0 \leq a \leq x}\{\ell(a)
            + \beta w^*(x-a)\}$ is single-valued and satisfies
            \begin{equation}\label{eq:EN_gen}
                (w^*)'(x) = \ell'(\pi^*(x))
                \qquad (0 < x < \hat x).
            \end{equation}
    \end{enumerate}
\end{proposition}

Stability of the fixed point 
is derived from the monotonicity and concavity of the Bellman
operator in Appendix~\ref{ss:model-proofs}.
In Proposition~\ref{p:finite} we show that the convergence $T^nw \to w^*$
always converges in finite time.

\subsection{Equivalence}
So far, we have solved the Bellman equation~\eqref{eq:epe} and derived
properties of its solutions. However, it is not clear whether the Bellman
equation can characterize the solution to the dynamic optimization
problem~\eqref{eq:dpm}, since the constraint $\sum_t a_t = \hat{x}$ is not
in the Bellman equation. We turn to this issue now.

Let 
\begin{equation}
    \label{eq:vf}
    W(x) 
    := \min \; \left\{
        \sum_{t=0}^{\infty} \beta^t \ell(a_t)
        \,\; : \;\,
        \{a_t\} \in \RR_+^{\infty} 
        \text{ and }
        \sum_{t=0}^{\infty} a_t = x
        \right\}
\end{equation}
be the \emph{value function} of the optimization problem~\eqref{eq:dpm}.
The next proposition shows that $W=w^*$, the fixed point of $T$,
and that the policy correspondence $\pi^*$
solves~\eqref{eq:dpm}. The proof can be found in
Appendix~\ref{ss:model-proofs}.

\begin{proposition} \label{p:policy}
    The sequence $\{a_t^*\}$
    defined by $x_0 = \hat x$, $x_{t+1} = x_t - \pi^*(x_t)$ and $a_t^* =
    \pi^*(x_t)$ is the unique solution to \eqref{eq:dpm}. Moreover, $W =
    w^*$.
\end{proposition}

The envelope condition \eqref{eq:EN_gen} now evaluates to
\begin{equation}
    \label{eq:EN}\tag{EN}
    W'(x_t) = \ell'(a_t^*)
\end{equation}
for all $t\geq 0$, which links marginal value to marginal disutility at
optimal action. Furthermore, \eqref{eq:EN} implies that the sequence
$\{a_t^*\}$ satisfies\footnote{To see this, note that $a^*_t$ solves
  $\inf_{0 \leq a \leq x_t} \left\{ \ell(a) + \beta w^*(x_t - a) \right\}$.
  Since both $\ell$ and $w^*$ are convex, elementary arguments show that
  either $\ell'(a^*_t) = \beta(w^*)'(x_t - a^*_t)$ or $a^*_t = x_t$. It
  follows from \eqref{eq:EN} that either $\ell'(a^*_t) = \beta
  \ell'(a^*_{t+1})$ or $a^*_{t+1} = 0$, which is equivalent to
  \eqref{eq:EU}.}
\begin{equation}
    \label{eq:EU}\tag{EU}
    \ell'(a_{t+1}^*) 
    = \max 
    \left\{
        \frac{1}{\beta} \ell'(a_t^*), \; \ell'(0)
    \right\}
\end{equation}
for all $t\geq 0$, which is akin to an Euler equation with a possibly
binding constraint. In the applications below we use \eqref{eq:EN} and
\eqref{eq:EU} to aid interpretation and provide economic intuition.

Equation~\eqref{eq:EU} implies that $\{a_t^*\}$ is a decreasing
sequence. This agrees with our intuition, since future losses are given greater
weight than current losses.

\subsection{Additional Results}

Instead of assuming $\ell' > 0$ as in \eqref{eq:ndas}, we can treat the case
$\ell'(0) = 0$, which has hitherto been excluded:

\begin{proposition}\label{p:co22}
    When $\ell'(0) = 0$, a feasible
    sequence $\{a^*_t\}$ solves \eqref{eq:dpm} if and only if \eqref{eq:EU}
    holds.  This sequence is unique, decreasing, and satisfies $a_t^* > 0$ for
    all $t$. 
\end{proposition}

Proposition~\ref{p:co22} shows that the Euler equation~\eqref{eq:EU}
becomes necessary and sufficient for
optimality when $\ell'(0)=0$. In fact, \eqref{eq:EU} can be reduced to
$\beta\ell'(a^*_{t+1}) = \ell'(a^*_t)$ in this case, which helps 
derive analytical solutions for some of the applications.

As the above results suggest, the set of tasks will be completed in finite
time if and only if $\ell'(0) > 0$.  The proof is in the appendix.

%Figure~\ref{fig:finite} provides
%an example, in which $T^n w$ converges to the fixed point $w^*$ in~5
%iterations.\footnote{In this example, $\ell(x) = e^{10x} - 1$, $f(x) =
%\ell'(0)x = 10x$, $\hat x = 1$, and $\beta = 2$.} 

%%
%\begin{figure}
  %\centering
  %\includegraphics[width=0.8\textwidth]{finite}
  %\caption{Convergence in finite steps when $\ell'(0) > 0$.}
  %\label{fig:finite}
%\end{figure}
%%

%\subsection{Extensions}

%There are potential extensions to the negative discount dynamic programming
%problem studied above. Here we suggest a few examples. Detailed analysis of
%each case is left for future research.

%\begin{example}
    %Let the disutility function $\ell$ and discount factor $\beta$ also
    %depend on the current state. Then the Bellman equation becomes
    %%
    %\begin{equation*}
        %f(x) = \min_{0 \leq a \leq x}
        %\left\{
            %\ell(x, a) + \beta(x) f(x-a)
        %\right\}.
    %\end{equation*}
    %%
%\end{example}

%\begin{example}
    %We can introduce another choice variable $t$ such that the Bellman
    %equation is
    %%
    %\begin{equation*}
        %f(x) = \min_{(a, t) \in G(x)}
        %\left\{
            %\ell(a, t) + \beta f(x-a, t)
        %\right\},
    %\end{equation*}
    %%
    %where $G(x)$ gives the set of available actions when facing state $x$.
    %The model in Section~\ref{ss:network} belongs to this case.
%\end{example}

%\begin{example}
    %Instead of assuming additively separable loss, we can define an
    %aggregator $W: \RR^2 \to \RR$ such that the Bellman equation becomes
    %%
    %\begin{equation*}
        %f(x) = \min_{0 \leq a \leq x} W \left(a, f(x-a) \right).
    %\end{equation*}
    %%
    %Then, \eqref{eq:epe} is a special case when $W(a, z) := \ell(a) + \beta
    %z$.
%\end{example}

\section{Application: Production Chains}\label{s:app:production}

Now we turn to applications of our negative discount dynamic program
motivated by production problems. We begin with linear production chains.

\subsection{A Coasian Production Chain}\label{ss:cpc}

In this section we consider a version of the Coasian models developed by
\cite{kikuchi2018span}, \cite{fally2018coasian} and \cite{yu2019equilibrium}.
We show how competitive equilibrium in these models can be
calculated using the dynamic programming theory from Section~\ref{s:model}.

\subsubsection{Set Up}

Consider a market with many
price-taking firms, each of which is either inactive or part of the
production of a single good.  To produce a unit of this good requires
implementing a unit mass of tasks.
The cost for any one firm of implementing an interval of 
length $v \in [0,1]$ is denoted $c(v)$, where $c$ is
increasing, strictly convex, continuously differentiable, and satisfies
$c(0)=0$.\footnote{Unlike \cite{kikuchi2018span}, we allow $c'(0)=0$.}

Firms face transaction costs, as 
a wedge between cost to the buyer and payment
received by the seller.\footnote{This follows \cite{kikuchi2018span} and also
    studies such as \cite{boehm2020misallocation}, where frictions in contract
    enforcement are treated as a variable wedge
    between effective cost to the buyer and payment to the supplier.}  
Transaction costs fall on buyers, so that, for a
transaction with face value $f$, the seller receives $f$ and the buyer
pays $(1+\tau) f$ with $\tau > 0$.\footnote{For example, $\tau f$ might be the cost of writing a contract
for a transaction with face value $f$.  This cost rises in $f$ because more
expensive transactions merit more careful contracts.}

Firms are indexed by integers $i \geq 0$.  A feasible allocation of tasks across firms is a
nonnegative sequence $\{v_i\}$ with $\sum_{i \geq 0} v_i
= 1$.  We identify firm $0$ with the most downstream
firm, firm $1$ with the second most downstream firm, and so on.  
Let $b_i$ be the downstream boundary of firm $i$, so
that $b_0 = 1$ and $b_{i+1} = b_i - v_i$ for all $i\geq 0$.
Then, profits of the $i$th firm are 
\begin{equation}
    \label{eq:prof}
    \pi_i = p(b_i) - c(v_i) - (1+\tau) p(b_{i+1}) . 
\end{equation}
Here $p \colon [0, 1] \to \RR_+$ is a price function, with
$p(t)$ interpreted as the price of the good at processing stage $t$.  

\begin{definition}
    \label{d:eq}
    Given a price function $p$ and a feasible allocation $\{v_i\}$,
    let $\{\pi_i\}$ be corresponding profits, as defined in
    \eqref{eq:prof}.  The pair $(p, \{v_i\})$ is called an
    \emph{equilibrium} for the production chain if
    \begin{itemize}
        \item[1.] $\; p(0) = 0$,
        \item[2.] $p(s) - c(s - t) -  (1+\tau) p(t) \leq 0$ for any pair $s, t$
            with $0 \leq t \leq s \leq 1$, and
        \item[3.] $\; \pi_i = 0$ for all $i$.
    \end{itemize}
\end{definition}
Conditions 1--3 eliminate profits for active firms and prevent entry by
inactive firms.

\subsubsection{Solution by Dynamic Programming}

\label{sss:sdp}

An equilibrium of the production chain satisfies
$p(b_i) = c(v_i) + (1+\tau)p(b_i - v_i)$, which has the same form as the
Bellman equation~\eqref{eq:epe}. Moreover, iterating on this relation yields
the price of the final good
\begin{equation}
    \label{eq:p0}
    p(1) = \sum_{i \geq 0} (1+\tau)^i c(v_i),
\end{equation}
which is analogous to the total loss in \eqref{eq:dpm}.
These facts lead us to a version of the negative discount
dynamic program introduced in Section~\ref{s:model} where a (fictitious)
agent seeks to minimize $\sum_{i \geq 0} (1 + \tau)^i c(a_i)$ subject to
$\sum_{i \geq 0} a_i = 1$.  By construction, any
feasible action path is also a feasible allocation of tasks in the
production chain.

Since $\ell = c$ and $\beta = 1 + \tau$,
the assumptions in~\eqref{eq:ndas} are satisfied.  Hence 
there exists a unique solution $\{a^*_i\}$ by Proposition~\ref{p:policy}. Let $W$ be the corresponding
value function given by \eqref{eq:vf}. The next proposition shows that the
solution to this dynamic program is precisely the competitive equilibrium of
the Coasian production chain.

\begin{proposition}
    \label{p:oie}
    If $p = W$ and $v_i = a_i^*$ for all $i \geq 0$, then $(p, \{v_i\})$
    is an equilibrium for the production chain.
\end{proposition}

For firm
with downstream boundary $b_i$, the envelope condition~\eqref{eq:EN} yields
\begin{equation}
    \label{eq:dofc}
    p'(b_i) = c'(v_i).
\end{equation}
Since $v_i$ is the optimal range of tasks implemented in-house by firm~$i$
in equilibrium, this expresses Coase's key idea: the size of the
firm is determined as the scale that equalizes the marginal costs of in-house
and market-based operations.
The Euler equation~\eqref{eq:EU} also implies
that $\{v_i\}$ is decreasing, so firm size increases with
downstreamness. This generalizes a finding of
\cite{kikuchi2018span}.

\subsubsection{An Example}

\label{sss:cf}

Suppose that the range of
tasks $v$ implemented by a given firm satisfies $v = f(k, n)$, where $k$ is
capital and $n$ is labor.  Given rental rate $r$ and wage rate $w$, the cost
function is $c(v) := \min_{k, n} \{rk + wn \}$ subject to $f(k, n) \geq v$.
Suppose further that, as in \cite{lucas1978size}, the production
function has the form $\phi(g(k, n))$, where $g$ has constant returns to scale
and $\phi$ is increasing and strictly concave (due to
``span-of-control'' costs).  To generate a closed-form solution, we take
    $g(k, n) = A k^\alpha n^{(1-\alpha)}$ and 
    $\phi(x) = x^\eta$,
with $0 < \alpha, \eta < 1$.  The resulting cost function has the form $c(v) =
\kappa v^{1/\eta}$, where $\kappa$ is a positive constant.

By Proposition~\ref{p:oie}, the optimal action path for the fictitious agent
corresponds to the equilibrium allocation of tasks across firms, and the
value function is the equilibrium price function. Since $c'(0) = 0$,
Proposition~\ref{p:co22} applies and~\eqref{eq:EU}
yields $a^*_{i+1} = \theta a^*_i$ for all $i \geq 0$, where $\theta :=
(1+\tau)^{\eta/(\eta - 1)} < 1$. From $\sum_{i = 0}^{\infty} a^*_i = 1$ we
obtain $v_i = a^*_i = \theta^i (1 - \theta)$.  Substituting this path into
\eqref{eq:vf} gives the price function
\begin{equation}\label{eq:wcfe}
    p(x) = W(x) = 
    \kappa 
    \left(1 - \theta\right)^{(1-\eta)/\eta} x^{1/\eta}.
\end{equation}
As anticipated by the theory, $p$ is strictly increasing and strictly
convex.

Intuitively, firm-level
span-of-control costs cannot be eliminated in aggregate due to transaction
costs, which force firms to maintain a certain size. This leads to strict
convexity of prices. If firms have constant returns to management ($\eta=1$),
then the price function in \eqref{eq:wcfe} becomes linear.\footnote{The above
    result on the size of firms is related to 
    \cite{antras2020geography}, who prove it is optimal to locate
    relatively downstream stages of production in relatively central locations
    where trade costs are lower. Their result holds because trade costs have more
    pronounced effects in more downstream stages of production in their model.
    Similarly, in our model, transaction costs have more pronounced effects in
    more downstream states of production, due to~\eqref{eq:EU}.}

%We can introduce countries that can be ranked 
%in terms of their transaction costs so that firms can choose in which 
%country to produce. Then, our model would predict that more downstream 
%stages of production should take place in countries with lower transaction costs 
%allowing us to speak to comparative advantage.

\subsection{Specialization and Failure Probabilities}

Production processes typically involve a series of complementary tasks. 
Mistakes in any one task can dramatically reduce the product's value.
Implications of such specialization and failure probabilities were studied in,
among others, the O-ring theory of economic development by
\cite{kremer1993ring} and the production chain models of
\cite{levine2012production} and \cite{costinot2013elementary}.  
These papers show how equilibrium
allocations can serve to mitigate the potentially exponential cost of failures in
long production chains.\footnote{For example, in \cite{levine2012production}, long chains
    involve a high degree of specialization and produce a large quantity of
    output but are also more prone to failure. However, chains in his model
    are long only if the failure rate is low thus mitigating the exponential
    impact that production failure of a single link has on output.  Similarly,
    \cite{costinot2013elementary},  in a global supply chain model where
    production of the final goods is sequential and subject to mistakes, show
that countries with lower probabilities of making mistakes specialize in later
stages of production.} In this section, we show that these ideas 
 are also amenable to analysis
using the negative discount dynamic program from Section~\ref{s:model}.

Consider, as before, a competitive market where producers implement a
mass of tasks contained in $[0, 1]$. We drop the assumption of
positive transaction costs and replace it with 
positive probability of defects.\footnote{Defects can alternatively be understood as
iceberg costs, where some percentage of goods are lost in transporting them
from one producer to the next.} Due to these defects, a producer who buys at
stage $t$ and sells at $s > t$ must buy $1+\tau$ units of the partially
completed good at $t$ to sell one unit of the processed good at $s$. Larger
$\tau$ then corresponds to a production process that is more prone to
failure. Profits for such a firm facing price function $p$ are
\begin{equation*}
	\pi = p(s) - c(s-t) - (1 + \tau) p(t).
\end{equation*}
This parallels the profit function \eqref{eq:prof} from the Coasian
production chain model.  If we
adopt the Cobb--Douglass technology from Section~\ref{sss:cf}, 
then the price of the final good is
\begin{equation}
	\label{eq:final-price}
	p^*(1) = \kappa \left( 1- (1 + \tau)^{\eta/(\eta-1)} \right)^{(1-\eta)/\eta}.
\end{equation} 

A rise in the failure probability
leads to only a moderate increase in the final good price. This is because
producers increase their range of internal production to mitigate any rise in cost
associated with a higher production failure of upstream producers. As a result,
there are fewer producers in production chains and the
compounding effect of higher production failures is limited. 

To clarify this point, let us compare this outcome with 
a hypothetical model where producers do not adjust their
production according to failure probabilities. Suppose in particular that
production chains are simply divided into equal tasks by $N$ 
producers. In this case, the final good price is
\begin{equation}
	\label{eq:final-price-hypothetical}
	{\hat p^*(1)}
	= \kappa \sum_{i=0}^N(1 + \tau)^i
	\left(\frac{1}{N}\right)^{1/\eta}
	= \kappa \frac{(1+\tau)^N - 1}{(1+\tau) - 1}
    \left( \frac{1}{N} \right)^{1/\eta}
	=O((1+\tau)^N).
\end{equation}
Now a small increase in $\tau$ increases the final good price
exponentially. This is intuitive, as an increase in cost compounds over all
producers involved in the production chain. See Figure~\ref{fig:price} for
a comparison of prices with and without producers adjusting for failure
probabilities.\footnote{In this example, we set $\kappa = 1$, $\eta = 0.5$,
	and $N = 50$.}

\begin{figure}
	\centering
	\includegraphics[width=0.8\textwidth]{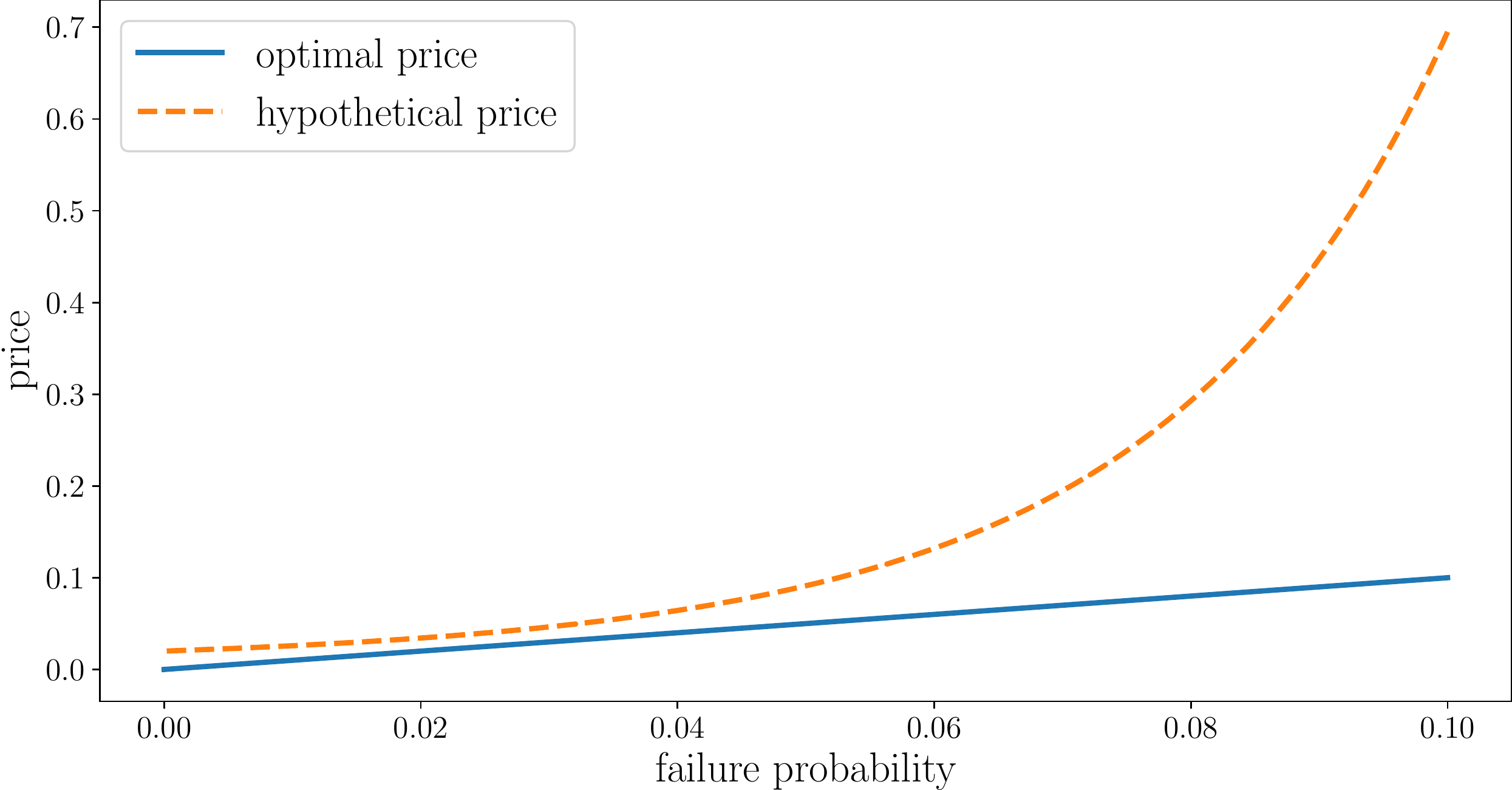}
	\caption{Final good price and failure probabilities.}
	\label{fig:price}
\end{figure}

Thus, returning to the original model, we see that equilibrium prices induce
producers to adjust to changes in failure probabilities, which optimally
mitigates the potentially exponential impact of failures on the cost of the
final good.

%\textcolor{blue}{
%Like
%transaction costs that affect inter-firm trade, these features tend to inhibit
%fine-grained division of tasks across firms.  In equilibrium, such costs are
%balanced against the gains from specialization. We show that the core ideas
%can also be analyzed via our framework. }

\section{Application: Knowledge and Communication}
\label{s:app:knowledge}

%\textcolor{blue}{
%which analyzes
%how workers are partitioned into classes such that each class is associated
%with a knowledge set to maximize output per capita. Optimal organizational
%forms are characterized by a pyramidal structure, in which workers are
%organized into layers with each layer smaller than the previous one.  We show
%how the same dynamic programming theory developed for production networks can
%be used to solve this internal organization problem.}

Many firms are characterized by a pyramidal structure, in which employees form
management layers with each layer smaller than the previous
one.  These features were modeled by \cite{garicano2000hierarchies}, 
where hierarchical organization of
knowledge involves a trade-off between the cost of acquiring problem solving
knowledge and the cost of communicating with others for help.   In this
section, we solve a version of Garicano's model using the dynamic programming
theory from Section~\ref{s:model}. 

Consider a firm where production requires solving a
set of problems.
Employees at management layer $i$ are
assigned problems $m_i\in[0,1]$. They learn to solve $z_i$ at cost $c(z_i)$ and pass
on the remainder $m_{i+1} = m_i - z_i$ to the next management layer $i+1$. This
incurs additional communication costs that are proportional to the value
of problems assigned to layer $i+1$ with coefficient $\tau$. 

Let $p\colon [0, 1] \to \RR$ be a (fictitious) price
function that assigns value to problems. Profits
of the $i$th management layer are 
\begin{equation*}
    \pi(m_i,z_i)=p(m_i)-(1+\tau)p(m_{i}-z_i)-c(z_i),
\end{equation*}
where $p(m_i)$ is the value of problems assigned to layer $i$, $(1+\tau) p(m_i -
z_i)$ is the cost of communicating and assigning unsolved problems to the
next layer, and $c(z_i)$ is the cost of learning to solve $z_i$.
Setting profits to zero and minimizing with respect to $m_{i+1}$ yields 
\begin{equation*}
    p(m_i)=\min_{m_{i+1}\leq m_i}\{c(m_i-m_{i+1})+(1+\tau)p(m_{i+1})\}.
\end{equation*}
This parallels the Bellman equation \eqref{eq:epe} of the negative discount
dynamic program in Section \ref{s:model}.

Suppose that $n$ employees can learn to solve $z = f(n)$ problems.
In other words, for a given
range of problems $z$, the number of employees required to solve $z$ is $n =f^{-1}(z)$.
Assume that $f$ is strictly increasing, strictly concave, and continuously
differentiable with $f(0) = 0$, and that $c(z) = wn = wf^{-1}(z)$ for some
wage rate $w$. Then the assumptions in \eqref{eq:ndas} are satisfied if we let $\ell = c$
and $\beta = 1+\tau$. The Euler equation~\eqref{eq:EU} implies that the
optimal sequence $\{z_i\}$ is decreasing, as is the number of
employees in each layer as $n_i = c(z_i)/w$. This replicates
Garicano's result that the top management layer has the smallest number of
employees and each layer below is larger than the one above.

The Euler equation \eqref{eq:EU} adds additional insight: each layer of
management acquires knowledge up to the point where the marginal cost of
learning equals the marginal cost of communicating and assigning unsolved
problems to the next layer.  The envelope condition~\eqref{eq:EN} implies
$p'(m_i)=c'(z_i)$, which says that, in equilibrium, the marginal value of
problems assigned to a management layer equals the marginal cost of learning
to solve problems within the layer.\footnote{This result is analogous to
    \eqref{eq:dofc} for the production chain model and reminiscent of Coase's
theory of the firm in the context of knowledge organization within a firm.}

\begin{figure}[tb!]
	\centering
	\begin{subfigure}{0.33\textwidth}
		\centering
		\includegraphics[height=0.9\textwidth]{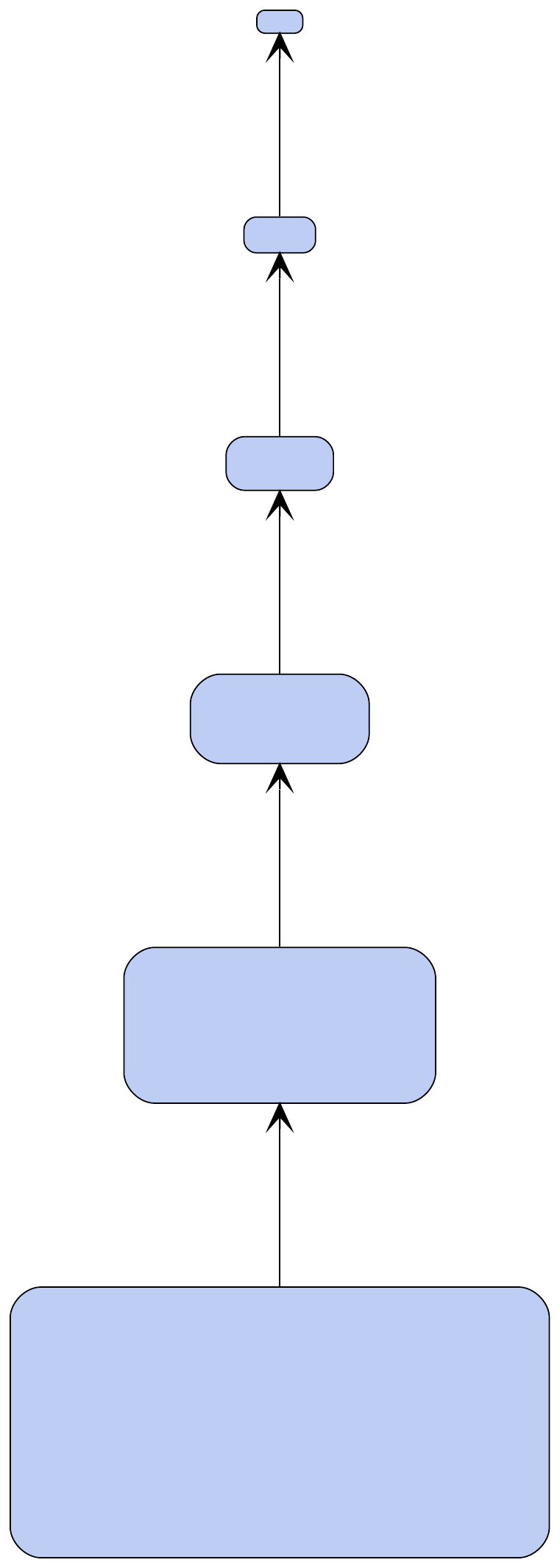}
		\caption{$\tau = 0.2$}
	\end{subfigure}%
	\begin{subfigure}{0.33\textwidth}
		\centering
		\includegraphics[height=0.9\textwidth]{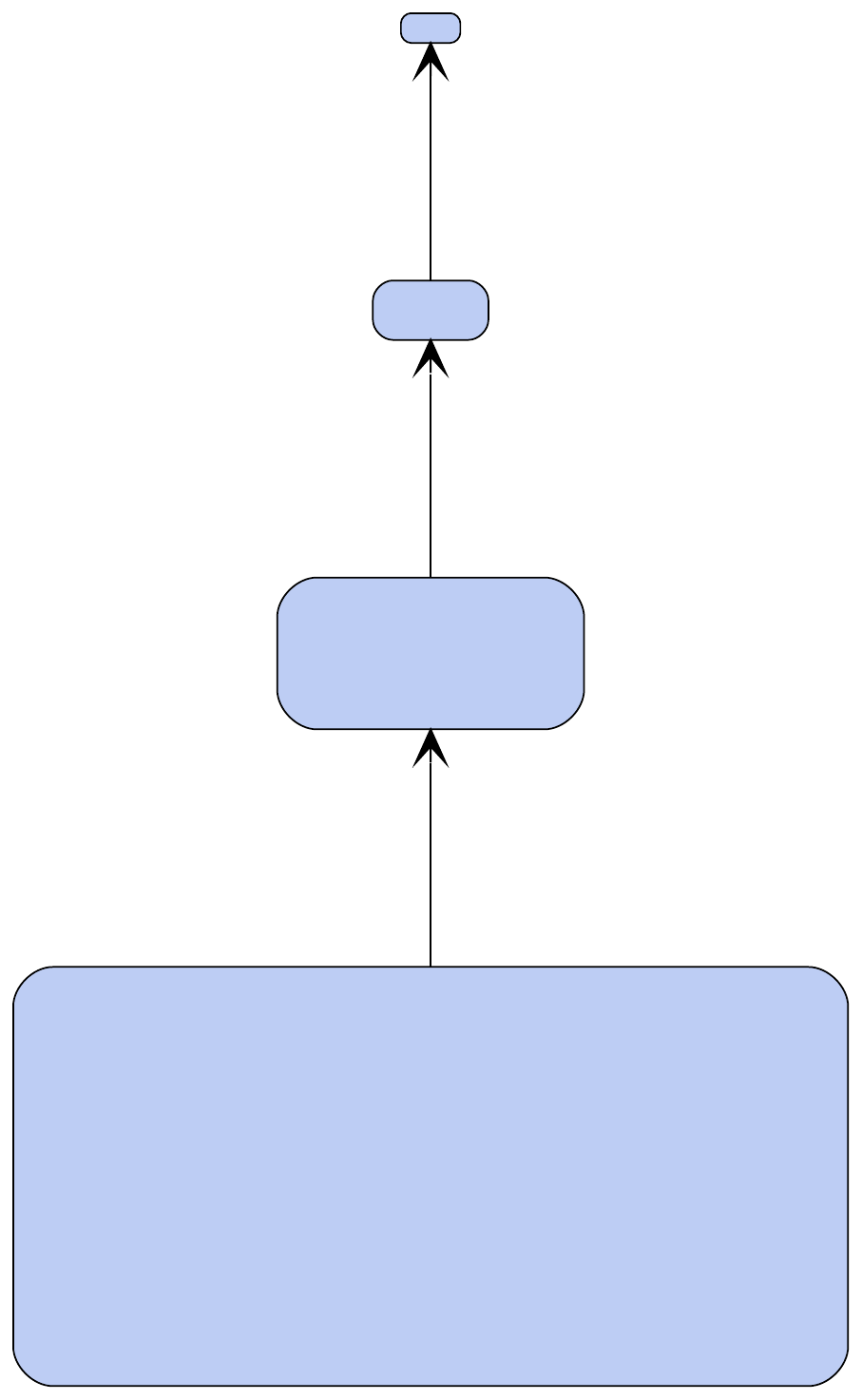}
		\caption{$\tau = 0.4$}
	\end{subfigure}
	\begin{subfigure}{0.33\textwidth}
		\centering
		\includegraphics[height=0.9\textwidth]{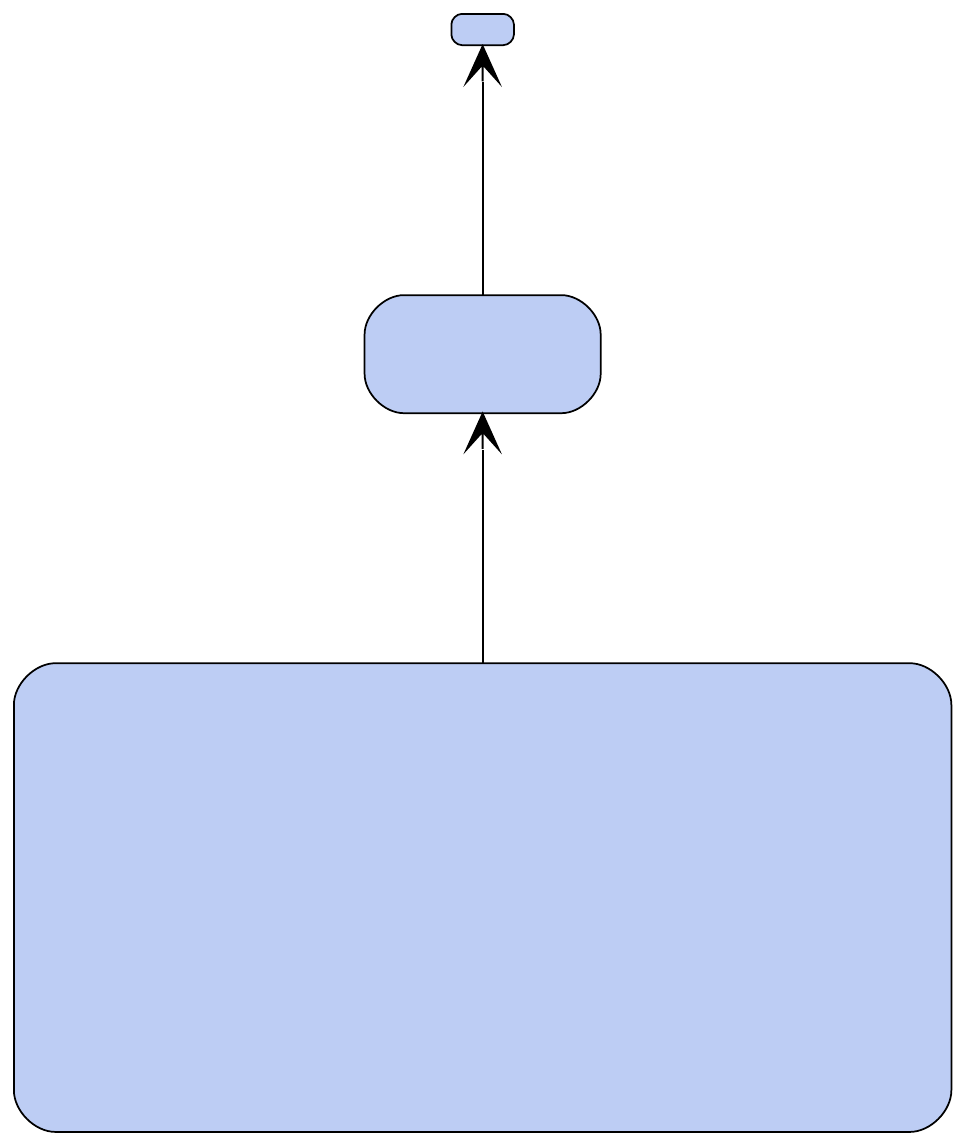}
		\caption{$\tau = 0.6$}
	\end{subfigure}%
	\caption{Optimal organizational structures.}
	\label{fig:garicano}
\end{figure}

Figure~\ref{fig:garicano} plots the optimal organizational structures of
three firms given by the model above.\footnote{We set $c(z) = z^{1.2}$ and
	$m_0 = 1$.} Each node corresponds to one management layer, who asks the
layer above for help, and its size is proportional to the number of
employees in that layer. As shown in the graphs, each firm has a pyramidal
structure and higher communication costs increase the relative knowledge
acquisition of lower layers and reduce the number of layers.

\section{Extension: Nonlinear Networks}
\label{s:app:networks}

In this section we treat more general network models that cannot be directly
handled by the theory in Section~\ref{s:model}.  Unlike the linear chains
discussed above, agents can interact with multiple partners. In
Section~\ref{ss:spatial}, we study a problem from economic geography.  In
Section~\ref{ss:network} we study chains with multiple upstream partners using
a general dynamic programming theory developed in Appendix~\ref{ss:gdp}.

\subsection{Spatial Networks}
\label{ss:spatial}

The distribution of city sizes shows remarkable regularity, as described by the
rank-size rule.\footnote{See \cite{gabaix2004evolution} and \cite{gabaix2009power} for
surveys.} One early attempt to  match the empirical city size
distribution is found in the central place theory of \cite{christaller1933central}.
\cite{hsu2012central} and \cite{hsu2014optimal} formalize Christaller's theory.
  In this section, we develop a model with similar insights 
 by extending our earlier dynamic programming results.

Consider a government that opens competition for many developers to build
cities to host a continuum of dwellers indexed by $[0, 1]$. Each developer can
build a large city that hosts everyone or  
build a smaller city and pay other developers to build 
``satellite cities'' that host the rest of the
population. Further satellites can be built
for existing cities until all dwellers are accommodated. 
This chain of city building
starts with a single developer, who is assigned the whole population, and
ends with a network of cities consisting of multiple layers.

Building satellite cities incurs extra costs that are charged as an ad
valorem tax on the payments to
the developers. We can think of the extra costs as costs of
providing public goods that connect different cities such as roads,
electricity, water, telecommunication, etc.  
Developers are paid according to a price $p \colon [0, 1] \to \RR$, which is a
function of the population assigned. The cost function of building a
city is $c \colon [0, 1] \to \RR$ and the tax rate is
$\tau$. A developer assigned to host $s$ dwellers maximizes profits by
solving
\begin{equation*}
    \max_{0 \leq t \leq s}
    \left\{
        p(s) - c(s-t) - (1+\tau) k p\left( \frac{t}{k} \right)
    \right\},
\end{equation*}
where $p(s)$ is the payment to the developer, $c(s-t)$ is the cost of
building a city of population $s-t$, $k$ is the number of satellite cities,
and $(1+\tau) k p(t/k)$ is the cost of assigning population $t/k$ to $k$
satellites.  In equilibrium, a city network is formed
where every dweller is accommodated and every developer makes zero profits.
The equilibrium price function satisfies
\begin{equation}
    \label{eq:city}
    p(s) = \min_{0 \leq t \leq s}
    \left\{
        c(s-t) + (1+\tau) k p\left( \frac{t}{k} \right)
    \right\}.
\end{equation}
To find the equilibrium price function, we first solve a negative discount
dynamic program and then show that its value
function is the solution to \eqref{eq:city}. 

Consider a dynamic optimization
problem with value function given by
\begin{equation}
    \label{eq:city_W}
    W(x) := \min_{\{v_i\}} \left\{
        \sum_{i=0}^\infty (1+\tau)^i k^i c(v_i):
        \{v_i\}\in \RR_+^\infty \text{ and } \sum_{i=0}^\infty k^i v_i = x
    \right\}.
\end{equation}
The problem in \eqref{eq:city_W} is a modified version of \eqref{eq:dpm}
that also features negative discounting and a convex loss function. In the
context of our city network model, \eqref{eq:city_W} describes a social
planner who minimizes the total cost of hosting the whole population, where
$v_i$ stands for the size of cities on layer $i$. 

In what follows we let $c(s) = s^\gamma$ with $\gamma > 1$. To emulate the
bifurcation process in \cite{hsu2012central} and \cite{hsu2014optimal}, we let
$k = 2$.  A similar argument to the proof of Proposition~\ref{p:co22} gives
the Euler equation
\begin{equation}
    \label{eq:EU_city}
    c'(v_i) = (1 + \tau) c'(v_{i+1}).
\end{equation}
Using this equation, it can be shown with some algebra that
$v_i = \theta^i (1-2\theta)$ for $\theta := (1+\tau)^{1/(1-\gamma)} < 1/2$
and the value function is $W(s) = (1-2\theta)^{\gamma-1} s^\gamma$.
It is
straightforward to verify that $p=W$ satisfies \eqref{eq:city}. Hence, the
value function for the social planner is also the equilibrium price function
under which all developers make zero profits.

The Euler equation~\eqref{eq:EU_city} describes the emergence of
optimal city hierarchy where each developer expands a city to accommodate
more dwellers until the marginal cost of expanding equals the marginal cost
of building and expanding satellite cities.
An envelope condition
similar to \eqref{eq:EN} also
holds: if a developer is assigned $s$ dwellers and delegate $t$ dwellers to
satellite cities, the equilibrium is reached when $p'(s) = c'(s-t)$. This
shows that the marginal value that a city provides must be equal to the
marginal cost of accommodating one more city dweller.

\begin{figure}
	\centering
	\includegraphics[width=0.8\textwidth]{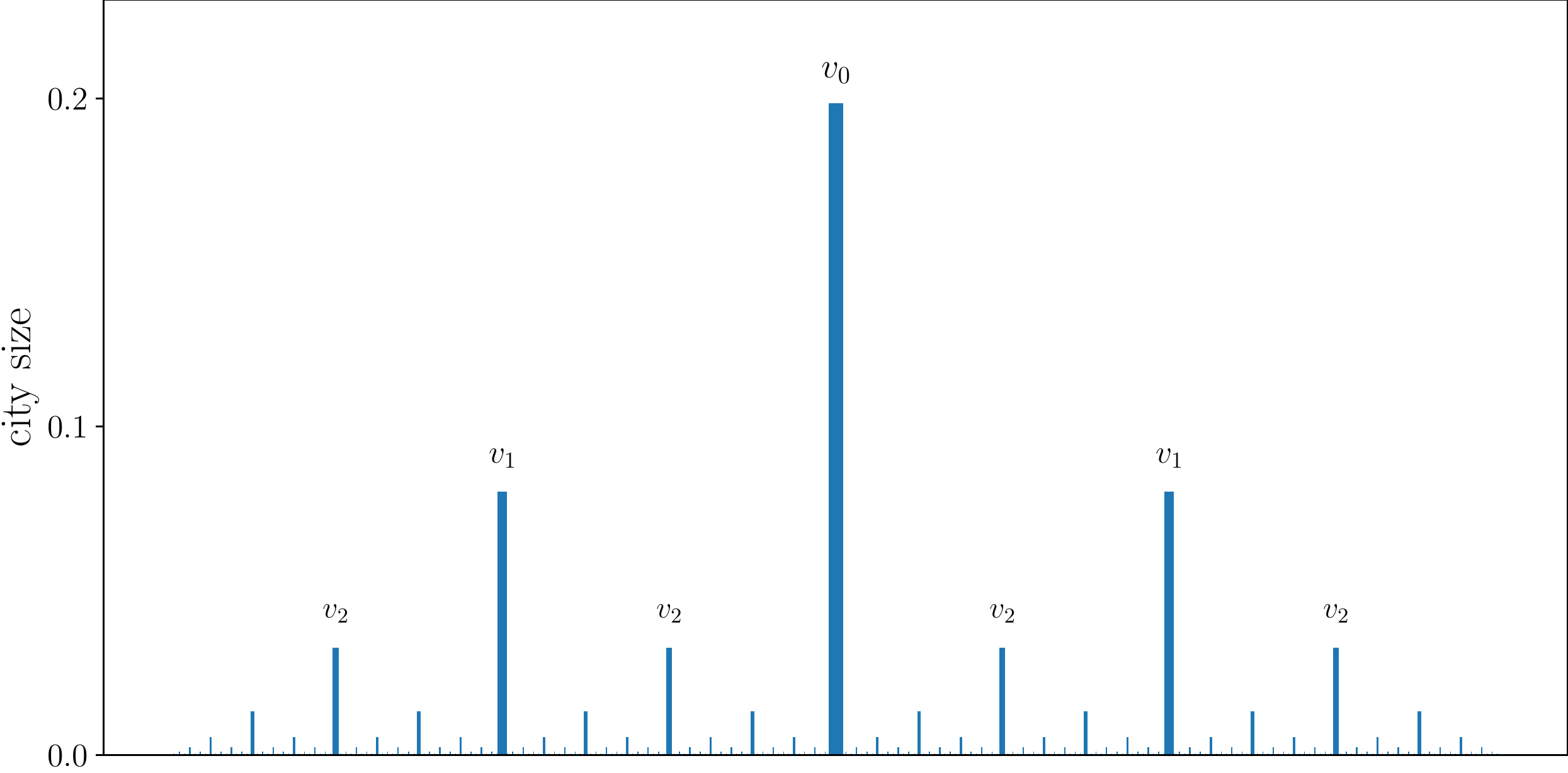}
	\caption{Illustration of optimal city hierarchy.}
	\label{fig:city}
\end{figure}

Figure~\ref{fig:city} illustrates the optimal city hierarchy by placing
cities according to \cite{hsu2012central} and
\cite{hsu2014optimal}.\footnote{We set $\gamma = 1.2$ and $\tau = 0.2$.} It
replicates the relative sizes of cities on different layers as in
\cite{hsu2012central} and \cite{hsu2014optimal}. Moreover, since the number
of cities doubles from one layer to the next, the rank of a city on layer
$i$ is $2^i$. Hence, the city size distribution generated by our model
follows a power law similar to \cite{hsu2012central}. In fact, the rank and
size of a city satisfy
\begin{equation*}
    \ln(Rank) = - \frac{\ln(1/2)}{\ln(\theta)} \ln(Size) + C,
\end{equation*}
where $C$ is a constant determined by $\theta$. When $\theta$ approaches
$1/2$, the slope approaches one, which corresponds to the well-documented
rank-size rule.

\subsection{Snakes and Spiders}\label{ss:network}

Modern production networks are characterized by processes that are both
sequential and non-sequential.
\cite{baldwin2013spiders} refer to the sequential processes  as ``snakes'' and
the non-sequential processes as ``spiders'', and analyze how the location of
different parts of a production chain is determined by unbundling costs of
production across borders. Here we study a model of production
networks featuring both snakes and spiders. 

As in \cite{kikuchi2018span} and \cite{yu2019equilibrium}, we consider a
generalization of the production chain model in Section~\ref{ss:cpc}, where
each firm can also choose the number of suppliers.  
 To account for costs of
extending spiders, we assume that firms bear an additive assembly cost $g$
that is strictly increasing in the number of suppliers, with $g(1) = 0$. Then
for a firm at stage $s$ that subcontracts tasks of range $t$ to $k$ suppliers,
the profits are
\begin{equation*}
p(s) - c(s-t) - g(k) - (1+\tau)kp(t/k),
\end{equation*}
where $p$ is the price function. Having multiple suppliers leads to another
trade-off: firms potentially benefit from subcontracting at a lower price
but also have to pay additional assembly costs.

We index the layers in the production network by integers $i\geq 0$ with layer 0
consisting only of the most downstream firm. Let $b_i$ be the downstream
boundary of firms on layer $i$, each producing $v_i$ and having $k_i$
suppliers. Then the boundary of firms on the next layer is given by $b_{i+1}
= (b_i - v_i)/k_i$.  We call $(p, \{v_i\},
\{k_i\})$ an equilibrium for the production
network if (i) $p(0) = 0$, (ii) $p(s) - c(s-t) - g(k) - (1+\tau)kp(t/k) \leq
0$ for all $0 \leq t \leq s \leq 1$ and $k\in\NN$, and (iii) $\pi_i = 0$ for
all $i\geq 0$ where
\begin{equation}\label{eq:prof_mp}
    \pi_i := p(b_i) - c(v_i) - g(k_i) - (1+\tau)k_ip
    \left( \frac{b_i - v_i}{k_i} \right).
\end{equation}

As in Section~\ref{sss:sdp}, we seek to find an equilibrium using dynamic
programming methods. Let $p^*$ be the solution to the following Bellman
equation
\begin{equation}\label{eq:mp}
    p(s) = \min_{\substack{0 \leq t\leq s\\ k\in\NN}}
    \left\{
    c(s-t) + g(k) + (1+\tau)kp(t/k)
    \right\}.
\end{equation}
Let $v_i = b_i - t^*(b_i)$ and $k_i = k^*(b_i)$ where $t^*(s)$ and $k^*(s)$
are the minimizers under $p^*$. Let $\iI$ be all continuous $p$ such that
$c'(0)s \leq p(s) \leq c(s)$ for all $s\in [0, 1]$.

\begin{proposition}
	\label{p:eq_mp}
	If $c'(0) > 0$ and $g(k) \to \infty$ as
	$k\to\infty$, then \eqref{eq:mp} has a unique solution $p^*\in\iI$ and $(p^*,
	\{v_i\}, \{k_i\})$ is an equilibrium for the production network.
\end{proposition}

In Appendix~\ref{ss:networks-proofs}, we show that the unique solution $p^*$ can be
computed by value function iteration. We then prove that $p^*$ induces an
equilibrium allocation. Theorem~\ref{t:con} can also be used to show the
monotonicity of $p^*$.

Figure~\ref{fig:network} plots two production networks with different
transaction costs, where each node corresponds to a firm in the network and
the one in the center is the most downstream firm.\footnote{We set $c(v) =
    v^{1.5}$ and $g(k) =
  0.0001(k-1)^{1.5}$.} The size of each node is proportional to the size of
    the firm, represented by the sum of assembly and transaction costs.
    Figure~\ref{fig:network} shows that more downstream firms are larger and
    have more upstream suppliers. Comparing panels (A) and (B), we can see
    that lower transaction costs increase the number of firms involved in the
    production network, encouraging the expansion of snakes. This is in line
    with the model prediction of \cite{baldwin2013spiders} that decreasing
    frictions leads to a finer fragmentation of the
    production.\footnote{\citeauthor{tyazhelnikov2019production}'s \citeyearpar{tyazhelnikov2019production} model of international
        production chains shares some features with the model above.  His
        model nests both snakes and spiders.  Each firm makes optimal decision
        conditional on its production location at the next stage. If we
        interpret market transactions as offshoring, the multiple upstream
        supplier model becomes a model in which firms decide to produce parts
    of a production chain in any number of countries.}

      %\footnote{The model developed in this section also bears some 
%similarities with the
%endogenous production network model in \cite{acemoglu2020endogenous}, where it
%is found that reducing distortions associated with an ad valorem tax leads to denser production
%networks. A reduction of transaction costs in our model captures the same idea, as shown
%in Figure~\ref{fig:network}.}

 %\textcolor{red}{[The reason why we don't refer to any specific equations to 
 %relate to the results by \cite{baldwin2013spiders}  and \cite{tyazhelnikov2019production} is 
%because we don't have comparative statics results in the multiple partner model. 
%If we have at least a result that reduction of transaction costs leads to smaller cope of firms,
%we could refer to that result here. I remember Junnan mentioned this to me before. 
%Alternatively, we could say we obtain 
%a similar prediction in our production line model. The issue here is that 
%\cite{baldwin2013spiders}  and \citeyearpar{tyazhelnikov2019production} develop more complex models
%but we seem to able to capture their results in our less complex model.] 
%}

\begin{figure}[tb!]
  \centering
  \begin{subfigure}{0.5\textwidth}
    \centering
    \includegraphics[height=0.7\textwidth]{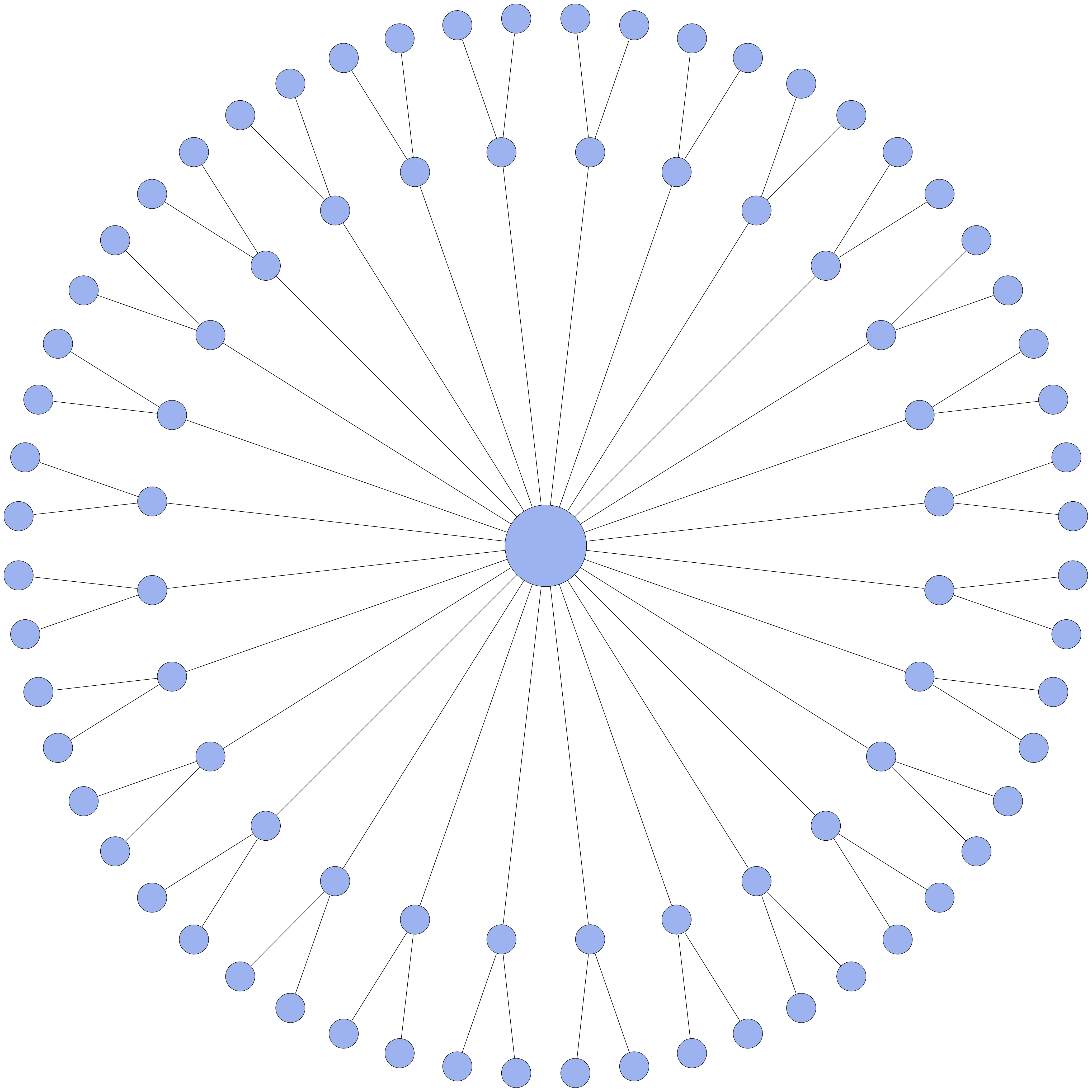}
    \caption{$\tau = 0.2$}
  \end{subfigure}%
  \begin{subfigure}{0.5\textwidth}
    \centering
    \includegraphics[height=0.7\textwidth]{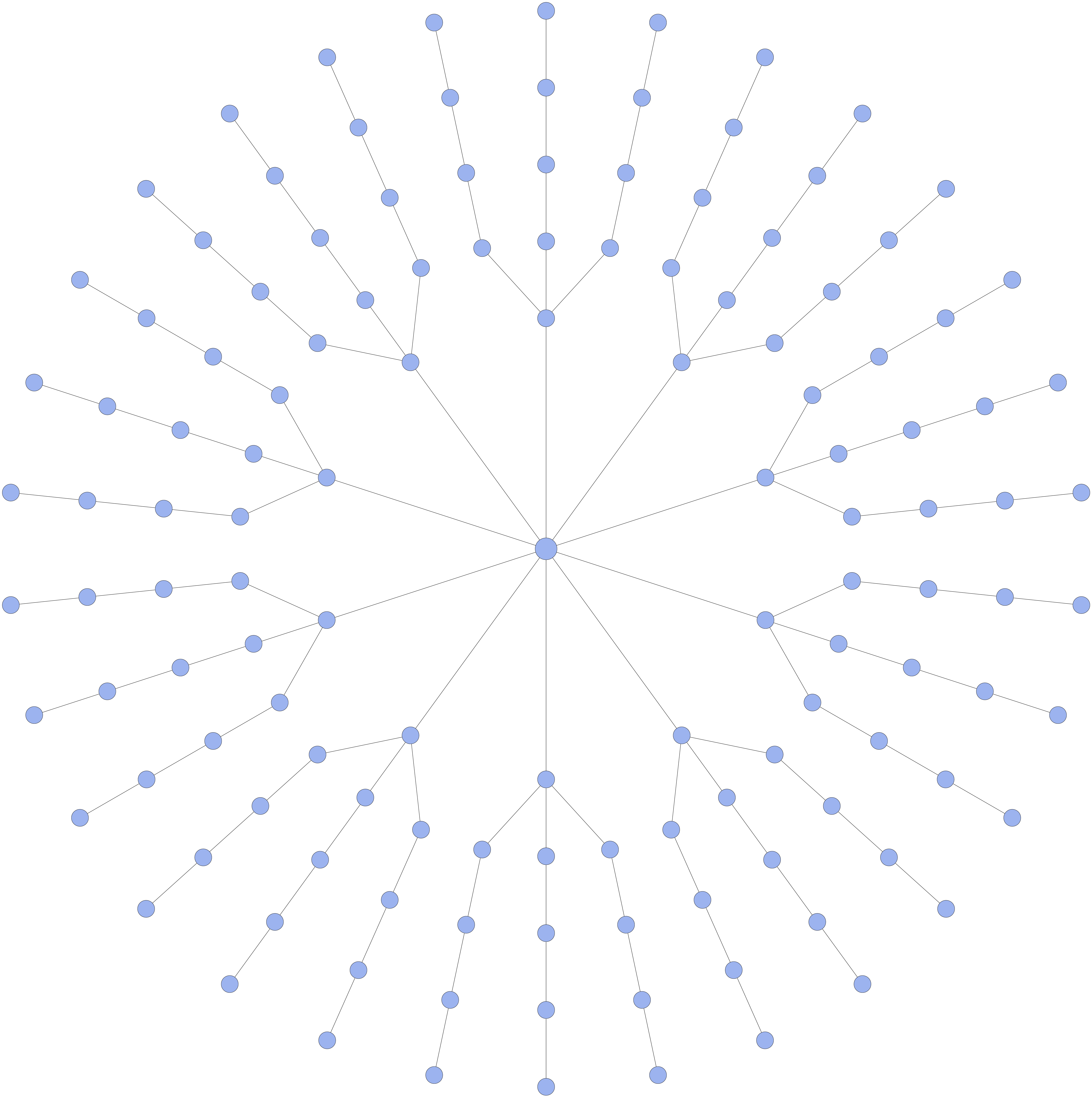}
    \caption{$\tau = 0.05$}
  \end{subfigure}
  \caption{Examples of production networks.}
  \label{fig:network}
\end{figure}

% In addition, we obtain sharp predictions on how the size of firms and the
% number of upstream suppliers increase at more downstream stages of
% production.

\section{Conclusion}

\label{s:c}

This paper shows how competitive equilibria in a range of production chain and
network models can be recovered as solutions to dynamic programming problems.
Equilibrium prices are identified with the value function of a dynamic
program, while competitive allocations across firms are identified with
choices under the optimal policy.  Dynamic programming methods are brought to
bear on both the theory of the firm and the structure of production networks,
providing new insights, as well as new analytical and computational methods.
In addition to production problems, we also consider related competitive
problems from economic geography and firm management.

Apart from the model of snakes and spiders in Section~\ref{ss:network}, all of
the problems faced by individual firms are convex.  This assumption allowed us to obtain
sharp results and useful characterizations.  An important remaining task is to
extend our results to a range of cases that feature non-convexities. This work
is left for future research.

\appendix
\section{Appendix}

\subsection{A General Dynamic Programming Framework}\label{ss:gdp}

In this section, we provide a general dynamic programming framework suitable
for analyzing equilibria in production networks. 

\subsubsection{Set Up}

Given a metric space $E$, let $\RR^E$ denote the set of functions from $E$
to $\RR$ and let $c\RR^E$ be all continuous functions in $\RR^E$. Given $g,
h \in \RR^E$, we write $g \leq h$ if $g(x) \leq h(x)$ for all $x \in E$, and
$\| f\| := \sup_{x \in E} |f(x)|$.

Let $X$ be a compact metric space. Let $A$
be a metric space and let $G$ be a nonempty, continuous, compact-valued
correspondence from $X$ to $A$. We understand $G(x)$ as the set of
available actions $a \in A$ for an agent in state $x$. Let $F_G :=
\setntn{(x,a)}{x \in X,\; a \in G(x)}$ be all feasible state-action pairs.
Let $L$ be an aggregator, mapping $F_G \times
\RR^X$ into $\RR$, with the interpretation that $L(x, a, w)$ is lifetime
loss associated with current state $x$, current action $a$ and continuation
value function $w$.
A pair $(L, G)$ with these properties is referred to as a
\emph{dynamic program}. 

The Bellman operator associated with
such a pair is the operator $T$ defined by
\begin{equation}\label{eq:genbe}
    (Tw)(x) = \inf_{a \in G(x)} L(x, a, w)
    \qquad (w \in \RR^X, \; x \in X).
\end{equation}
A fixed point of $T$ in $\RR^X$ is said to satisfy the Bellman
  equation.

\subsubsection{Fixed Point Results}\label{sss:fp}

Fix a dynamic program $(L, G)$ and consider the following
assumptions:

\begin{itemize}
    \item[$A_1$.] $(x,a) \mapsto L(x, a, w)$ is continuous on $F_G$
        when $w \in c\RR^X$.
    \item[$A_2$.] If $u, v \in c\RR^X$ with $u \leq v$, then
        $L(x, a, u) \leq L(x, a, v)$ for all $(x,a) \in F_G$.
    \item[$A_3$.] Given $\lambda \in (0, 1)$, $u, v \in c\RR^X$ and $(x,a) \in
        F_G$, we have
        \begin{equation*}
            \lambda L(x, a, u) + (1-\lambda) L(x, a, v)
            \leq L(x, a, \lambda u + (1-\lambda) v).
        \end{equation*}
    \item[$A_4$.] There is a $\psi$ in $c\RR^X$ such that $T\psi \leq \psi$.
    \item[$A_5$.] There is a $\phi$ in $c\RR^X$ and an $\epsilon > 0$ such that
        $\phi \leq \psi$ and $T\phi \geq \phi + \epsilon(\psi - \phi)$.
\end{itemize}

Assumptions $A_1$--$A_3$ impose some continuity, monotonicity and concavity.
Assumptions $A_4$--$A_5$ provide upper and lower bounds for the set of
candidate value functions.

Although contractivity is not imposed, we can show that the Bellman
operator \eqref{eq:genbe} is well behaved under $A_1$--$A_5$ after
restricting its domain to a suitable class of candidate solutions. To this
end, let
\begin{equation*}
    \iI := \setntn{ f \in c\RR^X}{\phi \leq f \leq \psi}.
\end{equation*}

\begin{theorem}\label{t:bk0}
    Let $(L, G)$ be a dynamic program and let $T$ be the Bellman
    operator defined in~\eqref{eq:genbe}. If $(L, G)$ satisfies
    $A_1$--$A_5$, then
    \begin{enumerate}
        \item $T$ has a unique fixed point $w^*$ in $\iI$.
        \item For each $w \in
            \iI$, there exists an $\alpha < 1$ and $M < \infty$ such that
            \begin{equation}
                \| T^n w - w^* \| \leq \alpha^n M
                \quad \text{ for all } \, n \in \NN.
            \end{equation}
        \item $\pi^* (x) := \argmin_{a \in G(x)} L(x, a, w^*)$ is upper
            hemicontinuous on $X$.  
    \end{enumerate}
\end{theorem}

The fixed point results in Theorem~\ref{t:bk0} rely on the monotonicity and
concavity of the Bellman operator. See Section~\ref{ss:gdp-proofs} for
details of the arguments and the proof of the theorem.

Theorem~\ref{t:bk0} does not discuss Bellman's principle of optimality. That
task is left until Section~\ref{sss:opt}. Regarding $\pi^*$, which has the
interpretation of a policy correspondence, an immediate corollary is that
$\pi^*$ is continuous whenever $\pi^*$ is single-valued on $X$.

\subsubsection{Shape and Smoothness Properties}

We now give conditions under which the solution to the Bellman equation
associated with a dynamic program possesses additional properties,
including monotonicity, convexity and differentiability. In what follows, we
assume that $X$ is convex in $\RR$ and $F_G$ is convex in $X \times A$. We
let
\begin{enumerate}
    \item $ic\RR^X$ be all increasing functions in $c\RR^X$ and
    \item $cc\RR^X$ be all convex functions in $c\RR^X$.
\end{enumerate}
We assume that $\iI$ defined above contains at least one element of each
set. The following assumption is needed for convexity and differentiability.

\begin{assumption}
    \label{a:condiff}
    In addition to $A_1$--$A_5$, the dynamic program $(L, G)$ satisfies the following conditions:
    \begin{enumerate}
        \item If $w \in cc\RR^X$,
            then $(x, a) \to L(x, a, w)$ is strictly convex on $F_G$.
        \item If $a \in \interior G(x)$ and $w \in cc\RR^X$, then $x \to L(x, a, w)$ is differentiable on $\interior X$.
    \end{enumerate}
\end{assumption}

We can now state the following result.

\begin{theorem}
    \label{t:con}
    If $Tw$ is strictly increasing for all $w\in ic\RR^X$, then $w^*$ is strictly increasing.
    If Assumption~\ref{a:condiff} holds, then 
    $w^*$ is strictly convex, $\pi^*$ is single-valued, $w^*$ is differentiable
    on $\interior X$ and
    \begin{equation}
        \label{eq:diff}
        (w^*)'(x) = L_x(x, \pi^*(x), w^*)
    \end{equation}
    whenever $\pi^*(x)\in \interior G(x)$. 
\end{theorem}

\subsubsection{The Principle of Optimality}\label{sss:opt}

If we consider the implications of the preceding dynamic programming theory,
we have obtained existence of a unique solution to the Bellman equation and
certain other properties, but we still lack a definition of optimal
policies, and a set of results that connect optimality and solutions to the
Bellman equation. This section fills these gaps.

Let $\Pi$ be all $\pi \colon X \to A$ such that $\pi(x) \in G(x)$ for all
$x\in X$. For each $\pi\in\Pi$ and $w\in \RR^X$, define the operator $T_\pi$
by
\begin{equation}
    \label{eq:Tpi}
    (T_\pi w)(x) = L(x, \pi(x), w).
\end{equation}
This can be understood as the lifetime loss of an agent following $\pi$ with
continuation value $w$. Let $\mM$ be the set of \emph{(nonstationary)
  policies}, defined as all $\mu = \{\pi_0, \pi_1, \ldots\}$ such that
$\pi_t \in \Pi$ for all $t$. For stationary policy $\{\pi, \pi, \ldots\}$,
we simply refer it as $\pi$. Let the \emph{$\mu$-value function} be defined
as
\begin{equation}
    \label{eq:piv}
    w_\mu(x) := \limsup_{n\to\infty} (T_{\pi_0} T_{\pi_1}\ldots T_{\pi_n} \phi)(x),
\end{equation}
where $\phi$ is the lower bound function in $\iI$.
Note that $w_\mu$ is always well
defined. The agent's problem is to minimize $w_\mu$ by choosing a policy in
$\mM$. The \emph{value function} $\bar{w}$ is defined by
\begin{equation}
    \label{eq:value}
    \bar{w}(x) := \inf_{\mu\in\mM} w_\mu(x)
\end{equation}
and the \emph{optimal policy} $\bar{\mu}$ is such that $\bar{w} =
w_{\bar{\mu}}$. We impose the following assumption.

\begin{assumption}\label{a:contbeta}
    In addition to $A_1$--$A_5$, the dynamic program $(L, G)$ satisfies the following conditions:
    \begin{enumerate}
        \item If $(x, a) \in F_G$, $v_n \geq \phi$ and $v_n \uparrow v$, then  $L(x, a, v_n) \to L(x, a, v)$.
        \item There exists a $\beta > 0$ such that, for all $(x, a) \in F_G$, $r > 0$ and $w \geq \phi$,
            \begin{equation}
                \label{eq:beta}
                L(x, a, w + r) \leq L(x, a, w) + \beta r    .
            \end{equation}
    \end{enumerate}
\end{assumption}

Part 1 of Assumption~\ref{a:contbeta} is a weak continuity requirement on the
aggregator with respect to the continuation value, similar to Assumption~4 in
\cite{bloise2018convex}. Part~2 of Assumption~\ref{a:contbeta} is analogous to
the Blackwell's condition, with the significant exception that $\beta$ in
\eqref{eq:beta} is not restricted to be less than one. 

\begin{theorem}
    \label{t:opt}
    If Assumption~\ref{a:contbeta} holds, then $w^* =
    \bar{w}$ and an optimal stationary policy exists. Moreover, a
    stationary policy $\pi$ is optimal if and only if $T_\pi \bar{w} = T
    \bar{w}$.
\end{theorem}

Theorem~\ref{t:opt} shows that the fixed point of the Bellman operator is
the value function and the Bellman's principle of optimality holds. It
immediately follows that any selector of $\pi^*$ in Theorem~\ref{t:bk0} is
an optimal stationary policy.

\subsection{Proofs for the General Theory}\label{ss:gdp-proofs}

To prove Theorem~\ref{t:bk0}, we first give a fixed point theorem for
monotone concave operators on a partially ordered Banach space due to
\cite{du1989fixed}.\footnote{The theory of monotone concave operators dates
  back to \cite{krasnoselskii1964positive}. Similar treatments include, for
  example, \cite{guo1988nonlinear}, \cite{guo2004partial}, and
  \cite{zhang2013variational}.}

\begin{theorem}[\citealp{du1989fixed}]
    \label{t:du}
    Let $P$ be a normal cone on a real Banach space $E$.\footnote{A cone
      $P\subset E$ is said to be normal if there exists $\delta>0$ such that
      $\|x + y\| \geq \delta$ for all $x, y\in P$ and $\|x\| = \|y\| = 1$.}
    Suppose $u_0, v_0 \in E$ with $u_0 < v_0$ and $A: [u_0, v_0] \to E$ is
    an increasing concave operator. If $Au_0 \geq u_0 + \epsilon (v_0 -
    u_0)$ for some $\epsilon \in (0, 1)$ and $Av_0 \leq v_0$, then $A$ has a
    unique fixed point $x^*$ in $[u_0, v_0]$. Furthermore, for any $x\in
    [u_0, v_0]$ and $n \in \NN$, $\|A^n x - x^*\| \leq M (1-\epsilon)^n$ for
    some $M > 0$.
\end{theorem}

\begin{proof}[Proof of Theorem~\ref{t:bk0}]
    By $A_1$ and Berge's theorem of the maximum, $Tw$ is continuous.  Hence $T$
    maps $c\RR^X$ to itself.  It follows directly from $A_2$ that $T$ is
    \emph{isotone} on $c\RR^X$, in the sense that $u \leq v$ implies $Tu \leq
    Tv$. Conditions $A_4$--$A_5$ and
    the isotonicity of $T$ imply that, when $\phi \leq w \leq \psi$,
    we have
        $\phi \leq T \phi \leq Tw \leq T \psi \leq \psi$.
    In particular, $T$ is an isotone self-map on $\iI$.

    The Bellman operator is also concave on $\iI$, in the sense that
    \begin{equation}\label{eq:defconcave}
        0 \leq \lambda \leq 1 \text{ and } u, v \in \iI \text{ implies }
        \lambda Tu + (1-\lambda) Tv \leq T(\lambda u + (1-\lambda) v).
    \end{equation}
    Indeed, fixing such $\lambda, u, v$ and applying $A_3$, we have
    \begin{equation*}
        \min_{a \in G(x)} 
        \left\{ 
            \lambda L(x, a, u) + (1-\lambda) L(x, a, v)
        \right\}
        \leq \min_{a \in G(x)} 
            L(x, a, \lambda u + (1-\lambda) v)
    \end{equation*}
    for all $x \in X$.  Since, for any pair of real valued functions $f,g$
    we have $\min_a f(a) + \min_a g(a) \leq \min_a \{f(a) + g(a)\}$, it follows
    that~\eqref{eq:defconcave} holds.

    The preceding analysis shows that $T$ is an isotone concave self-map on
    $\iI$. In addition, by $A_4$ and $A_5$, we have $T \psi \leq \psi$ and $T \phi \geq \phi +
    \epsilon (\psi - \phi)$ for some $\epsilon > 0$.  Since $\iI$ is an order
    interval in the positive cone of the Banach space $(c\RR^X, \| \cdot\|)$,
    and since that cone is normal and solid, the first two claims in
    Theorem~\ref{t:bk0} are now confirmed via Theorem~\ref{t:du}.
    The final claim is due to Berge's theorem of the maximum.
\end{proof}

\begin{proof}[Proof of Theorem~\ref{t:con}]
    The first part of the theorem follows directly from the fact that
    $ic\RR^X$ is a closed subspace. The proof is omitted.
    To prove the strict convexity of $w^*$, it suffices to show that $Tw$ is
    strictly convex for all $w \in cc\RR^X$ since $cc\RR^X$ is a closed subspace
    of $c\RR^X$. Pick any $x_1, x_2\in X$ with $x_1 < x_2$ and any $\lambda \in
    (0, 1)$. Let $x_\lambda = \lambda x_1 + (1-\lambda) x_2$. Pick any $w\in
    cc\RR^X$ and let $\pi_w \colon X \to A$ be such that $(Tw)(x) = L(x, \pi_w(x), w)$.
    It follows that
    \begin{align*}
        \lambda(Tw)(x_1) + (1 - \lambda)(Tw)(x_2)
        &= \lambda L(x_1, \pi_w(x_1), w) + (1 - \lambda)L(x_2, \pi_w(x_2), w)\\
        &> L(x_\lambda, \lambda \pi_w(x_1) + (1-\lambda) \pi_w(x_2), w)\\
        &\geq L(x_\lambda, \pi_w(x_\lambda), w) = (Tw)(x_\lambda),
    \end{align*}
    where the first inequality holds because $(x, a) \mapsto L(x, a, w)$ is
    strictly convex and the second inequality holds because $F_G$ is convex.
    Therefore, $w^*$ is strictly convex. Strict convexity of $L$ then
    implies that $\pi^*$ is single-valued.

    Since $\pi^*(x) \in
    \interior G(x)$ and $G$ is continuous, there exists an open neighborhood
    $D$ of $x$ such that $\pi^*(x) \in \interior G(y)$ for all $y \in D$.
    Define $W(y) := L(y, \pi^*(x), w^*)$ for all $y \in D$. Then $W(y) \geq
    w^*(y)$ for all $y\in D$ and $W(x) = w^*(x)$. Since $W$ is convex and
    differentiable on $D$, differentiability of $w^*$ and \eqref{eq:diff}
    then follow from \cite{benveniste1979differentiability}.
\end{proof}

We say that a dynamic programming problem has the \emph{monotone increasing}
property if $-\infty < \phi(x) \leq L(x, a, \phi)$ for all $(x, a) \in F_G$
and Assumption~\ref{a:contbeta} are satisfied. We state two
useful lemmas from \cite{bertsekas2013abstract}.

\begin{lemma}[Proposition 4.3.14, \cite{bertsekas2013abstract}]
    \label{l:MI}
    Let the monotone increasing property hold and assume that the sets
    \begin{equation*}
        G_k(x, \lambda) := \{x\in G(x) \mid L(x, a, T^k\phi) \leq \lambda\}
    \end{equation*}
    are compact for all $x\in X$, $\lambda \in \RR$, and $k$ greater than
    some integer $\bar{k}$. If $w\in \RR_+^X$ satisfies $\phi \leq w \leq
    \bar{w}$, then $\lim_{n\to\infty} T^n w = \bar{w}$. Furthermore, there
    exists an optimal stationary policy.
\end{lemma}

\begin{lemma}[Proposition 4.3.9, \cite{bertsekas2013abstract}]
    \label{l:opt}
    Under the monotone increasing property, a stationary policy $\pi$ is
    optimal if and only if $T_\pi \bar{w} = T\bar{w}$.
\end{lemma}

\begin{proof}[Proof of Theorem~\ref{t:opt}]
    Theorem~\ref{t:bk0} implies that $\lim_{n\to\infty}
    T^n \phi = w^*$. To prove $w^* = \bar{w}$, it suffices to show that the
    conditions of Lemma~\ref{l:MI} hold and $\phi \leq \bar{w}$.

    It follows from $A_5$ that $\phi(x) \leq (T \phi)(x) \leq L(x, a, \phi)$ for
    all $(x, a) \in F_G$. Therefore, the monotone increasing property is
    satisfied. Since $T$ is a self-map on $c\RR^X$, to check the conditions of
    Lemma~\ref{l:MI}, it suffices to prove that the set
    \begin{equation*}
        G(x, \lambda) := \{x\in G(x) \mid L(x, a, w) \leq \lambda\}
    \end{equation*}
    is compact for any $w\in c\RR^X$, $x\in X$, and $\lambda \in \RR$. Since $a
    \mapsto L(x, a, w)$ is continuous by $A_1$, $L(x, \cdot\,,
    w)^{-1}\left((-\infty, \lambda]\right)$ is a closed set. Since $G$ is
    compact-valued, $G(x, \lambda)$ is compact. It remains to show that $\phi
    \leq \bar{w}$. By $A_2$ and the monotone increasing property, we have for any $\mu =
    (\pi_0, \pi_1, \ldots) \in \mM$, $\phi \leq T_{\pi_0}
    T_{\pi_1}\ldots T_{\pi_n} \phi$ for all $n\in\NN$. Then by definition,
    $\phi \leq w_\mu$ for all $\mu \in \mM$. Taking the infimum gives $\phi
    \leq \bar{w}$.
    Lemma~\ref{l:MI} then implies that $w^* = \bar{w}$
    and there exists an optimal stationary policy. The principle of
    optimality follows directly from Lemma~\ref{l:opt}.
\end{proof}

\subsection{Proofs for Section~\ref{s:model}}\label{ss:model-proofs}

Let $\fF$ be the set of increasing convex functions in $\iI$. Throughout the
proofs, we regularly use the alternative expression for $T$ given by
\begin{equation}
    \label{eq:bellop2}
    (Tw)(x) = \min_{0 \leq y \leq x} \, \{ \ell(x - y) + \beta w(y) \}.
\end{equation}
Also, given $w \in \fF$, define
\begin{equation*}
    \pi_w(x) = \argmin_{0 \leq a \leq x} \, \{ \ell(a) + \beta w(x-a) \}
\end{equation*}
and
\begin{equation}
    \label{eq:deftell}
    \sigma_w(x) := \argmin_{0 \leq y \leq x} \{\ell(x - y) + \beta w(y) \}
    = x - \pi_w(x).
\end{equation}
These functions are clearly well-defined, unique and single-valued. Let
$\sigma = \sigma_{w^*}$ and $\pi = \pi_{w^*}$. Let $\eta$ be
the constant defined by
\begin{equation}
    \label{eq:eta}
    \eta := \max \, \setntn{0 \leq x \leq \hat x}{ \ell'(x) \leq \beta
      \ell'(0)}.
\end{equation}

We begin with several lemmas. The proof of the first lemma is trivial and
hence omitted.

\begin{lemma}
    \label{l:eta}
    We have $\eta > 0$ if and only if $\ell'(0) > 0$.  If $\eta <
    \hat x$, then $\ell'(\eta) = \beta \ell'(0)$.
\end{lemma}

\begin{lemma}
    \label{l:comp2}
    If $w \in \fF$, then $\sigma_w(x) = 0$ if and only if $x \leq \eta$.
\end{lemma}

\begin{proof}
    First suppose that $x \leq \eta$. Seeking a contradiction, suppose there
    exists a $y \in (0,x]$ such that $\ell(x - y) + \beta w(y) < \ell(x)$.
    Since $w \in \fF$ we have $w(y) \geq \ell'(0) y$ and hence
    \begin{equation*}
        \beta w(y) \geq \beta \ell'(0) y \geq \ell'(\eta) y.         
    \end{equation*}
    Since $x \leq \eta$, this implies that $\beta w(y) \geq \ell'(x) y$.
    Combining these inequalities gives $\ell(x - y) + \ell'(x) y < \ell(x)$,
    contradicting convexity of $\ell$.

    Now suppose that $\sigma_w(x)=0$. We claim that $x \leq \eta$, or,
    equivalently $\ell'(x) \leq \beta \ell'(0)$. To prove $\ell'(x) \leq
    \beta \ell'(0)$, observe that since $w \in \fF$ we have $w(y) \leq
    \ell(y)$, and hence
    \begin{equation*}
       \ell(x) 
        \leq \ell(x-y) + \beta w(y) 
        \leq \ell(x-y) + \beta \ell(y)
        \quad \text{for all } y \leq x.
    \end{equation*}
    It follows that
    \begin{equation*}
        \frac{\ell(x) - \ell(x - y)}{y} \leq \frac{\beta \ell(y)}{y}
        \quad \text{for all} \quad
        y \leq x.
    \end{equation*}
    Taking the limit gives $\ell'(x) \leq \beta \ell'(0)$.
\end{proof}

\begin{proof}[Proof of Proposition~\ref{p:ndp}]
	Let $A = X = [0, \hat x]$, $G(x) = [0,x]$ and $L(x, a, w) = \ell(a) +
    \beta w(x-a)$. Conditions $A_1$--$A_3$ in Section~\ref{sss:fp} obviously
    hold. Condition $A_4$ holds since $\min_{0 \leq a \leq x} \{ \ell(a) +
    \beta \ell(x-a) \} \leq \ell(x)$. For condition $A_5$, note that $L(x,
    a, \phi) = \ell(a) + \beta \ell'(0) (x-a)$. Then $T\phi = \ell$ if $x <
    \eta$ and $(T\phi)(x) = \ell(\eta) + \beta \ell'(0)(x - \eta)$ if $x
    \geq \eta$. For $x < \eta$, $T\phi - \phi = \psi - \phi$ so we can
    choose any $\epsilon \leq 1$. For $x \geq \eta$,
	\begin{align*}
	(T\phi)(x) - \phi(x) 
	&= \ell(\eta) + \beta \ell'(0) (x-\eta) - \ell'(0) x\\
	&= \ell(\eta) - \ell'(0)\eta + (\beta-1)\ell'(0)(x-\eta)\\
	&\geq \ell(\eta) - \ell'(0)\eta = (\psi - \phi)(\eta).
	\end{align*}
	Since $\psi - \phi$ is increasing, we can choose any $\epsilon \leq
    \bar{\epsilon}$ where $(\psi - \phi)(\eta) = \bar{\epsilon}(\psi -
    \phi)(\hat x)$. The first part of the proposition thus follows from
    Theorem~\ref{t:bk0}.

    Consider the alternative expression for $T$ in \eqref{eq:bellop2}. Since
    $\ell$ is strictly convex, $(x, y) \mapsto \ell(x - y) + \beta w(y)$ is
    strictly convex for all $w\in cc\RR^X$. Hence,
    part~1 of Assumption~\ref{a:condiff} holds. Evidently $Tw$ is strictly
    convex for all $w\in\fF$.

    Next we show that $Tw$ is strictly increasing for all $w \in \fF$. Pick
    any $w\in \fF$ and $x_1 \leq x_2$. For ease of notation, let $y_i =
    \sigma_w(x_i)$ for $i\in \{1, 2\}$. If $y_2 \leq x_1$, then
    \begin{align*}
        (Tw)(x_1) &= \ell(x_1 - y_1) + \beta w(y_1)\\
        &\leq \ell(x_1 - y_2) + \beta w(y_2)\\
        &< \ell(x_2 - y_2) + \beta w(y_2) = (Tw)(x_2),
    \end{align*}
    where the first inequality holds since $y_2$ is available when $y_1$ is
    chosen and the second inequality holds since $\ell$ is strictly
    increasing. If $y_2 > x_1$, we first consider the case of $x_1 + y_2 <
    x_2$. Then $(Tw)(x_2) > \ell(x_1) + \beta w(y_2) \geq \ell(x_1) \geq
    (Tw)(x_1)$. For the case of $x_1 + y_2 \geq x_2$, we have $0 \leq y_1'
    \leq x_1 < y_2$ where $y_1' = x_1 + y_2 - x_2$. Since $w$ is not
    constant, $w \in \fF$ implies that $w$ is strictly increasing. It
    follows that
    \begin{align*}
        (Tw)(x_1) &= \ell(x_1 - y_1) + \beta w(y_1)\\
        &\leq \ell(x_1 - y_1') + \beta w(y_1')\\
        &< \ell(x_2 - y_2) + \beta w(y_2) = (Tw)(x_2).
    \end{align*}
    Therefore, $T$ is a self-map on $\fF$ and $Tw$ is strictly increasing
    and strictly convex for all $w\in\fF$. Theorem~\ref{t:con} then implies
    that $w^*$ is strictly increasing and strictly convex.
    
    Since $\ell$ is differentiable, part~2 of Assumption~\ref{a:condiff}
    holds. Theorem~\ref{t:con} then implies that $w^*$ is differentiable and
    $(w^*)'(x) = \ell'(x - \sigma(x))$ whenever $\sigma(x)$ is interior.
    Lemma~\ref{l:comp2} implies that $w^*(x) = \ell(x)$ and thus $(w^*)'(x)
    = \ell'(x)$ when $x \leq \eta$; when $x > \eta$, $\sigma$ is interior
    and $(w^*)'(x) = \ell'(x - \sigma(x))$. Since $\sigma$ is continuous,
    $(w^*)'$ is continuous. Therefore, $w^*$ is continuously differentiable
    on $(0, \hat x)$ and $(w^*)'(x) = \ell'(\pi(x))$.
\end{proof}

The next lemma further characterizes $\pi$ and $\sigma$.

\begin{lemma}
    \label{l:tstarnz}
    Let $w \in \fF$. If $x_1, x_2$ satisfy $0 < x_1 \leq x_2$, then
    $\sigma_w(x_1) \leq \sigma_w(x_2)$ and $\pi_w(x_1) \leq \pi_w(x_2)$.
    Moreover, if $x \geq \eta$, then $\pi_w(x) \geq \eta$; if $x \leq \eta$,
    then $\pi_w(x) = x$.
\end{lemma}

\begin{proof}
    Pick any $w \in \fF$. Since $\ell$ and $w$ are convex, the maps $(x, a)
    \mapsto \ell(a) + \beta w(x-a)$ and $(x, y) \mapsto \ell(x-y) + \beta
    w(y)$ both satisfy the single crossing property. It follows from
    Theorem~$4'$ of \cite{milgrom1994monotone} that $\pi_w$ and $\sigma_w$
    are increasing.
    
    For the last claim, since $\pi_w$ is increasing, Lemma~\ref{l:comp2}
    implies that, if $\eta \leq x$, then $\pi_w(x) \geq \pi_w(\eta) = \eta -
    \sigma_w(\eta) = \eta$; and if $x \leq \eta$, then $\pi_w(x) = x -
    \sigma_w(x) = x$.
\end{proof}

The following lemma characterizes the solution to \eqref{eq:dpm} and is
useful when showing the equivalence between \eqref{eq:dpm} and
\eqref{eq:epe}.

\begin{lemma}
    \label{l:finn0}
    If $\{a_t\}$ is a solution to \eqref{eq:dpm}, then 
        $\{a_t\}$ is monotone decreasing and $a_{T+1} = 0$ if and only if $a_T \leq \eta$.
\end{lemma}

\begin{proof}
    The first claim is obvious, because if $\{a_t\}$ is a solution to
    \eqref{eq:dpm} with $a_t < a_{t+1}$, then, given that $\beta >
    1$, swapping the values of these two points in the sequence
    will preserve the constraint while strictly decreasing total loss.
    Regarding the second claim, since $\{a_t\}$
    is monotone decreasing, it suffices to check the case $a_T > 0$.  To this
    end, suppose to the contrary that $\{a_t\}$ is a solution to
    \eqref{eq:dpm} with $0 < a_T < \eta$ and $a_{T+1} > 0$.  Consider an
    alternative feasible sequence $\{\hat a_t\}$ defined by $\hat a_T = a_T +
    \epsilon$, $\hat a_{T+1} = a_{T+1} - \epsilon$ and $\hat a_t = a_t$ for
    other $t$.  If we compare the values of these two sequences we get
    \begin{align*}
        \sum_{t=0}^{\infty} \beta^t \ell(a_t) - \sum_{t=0}^{\infty} \beta^t
        \ell(\hat a_t)
        & = \beta^T [\ell(a_T) - \ell(a_T + \epsilon)]
            + \beta^{T+1} [\ell(a_{T+1}) - \ell(a_{T+1} - \epsilon)]
            \\
        & = \epsilon \beta^T 
            \left\{ 
            - \frac{\ell(a_T + \epsilon) - \ell(a_T)}{\epsilon}
            + \beta \frac{\ell(a_{T+1} - \epsilon) - \ell(a_{T+1})}{-\epsilon}
            \right\}.
    \end{align*}
    The term inside the parenthesis converges to 
    \begin{equation*}
        - \ell'(a_T) + \beta \ell'(a_{T+1}) > - \ell'(\eta) + \beta \ell'(0) \geq 0,
    \end{equation*}
    where the first inequality follows from $a_T \leq \eta$, $a_{T+1} > 0$ and
    strict convexity of $\ell$; and the second inequality is by the definition
    of $\eta$.  We conclude that for $\epsilon$ sufficiently small, the
    difference $\sum_{t=0}^{\infty} \beta^t \ell(a_t) - \sum_{t=0}^{\infty}
    \beta^t\ell(\hat a_t)$ is positive, contradicting optimality.

    Finally we check the claim $a_{T+1} = 0 \implies a_T \leq \eta$.
    Note that if $\eta = \hat x$ then there is nothing to prove, so we can and do
    take $\eta < \hat x$.  Seeking a contradiction, suppose instead that $a_{T+1} = 0$ and
    $a_T > \eta$.  Consider an alternative feasible sequence $\{\hat a_t\}$
    defined by $\hat a_T = a_T - \epsilon$, $\hat a_{T+1} = \epsilon$ and
    $\hat a_t = a_t$ for other $t$.  In this case we have
    \begin{equation*}
        \sum_{t=0}^{\infty} \beta^t \ell(a_t) - \sum_{t=0}^{\infty} \beta^t
        \ell(\hat a_t)
        = \epsilon \beta^T 
            \left\{ 
                \frac{\ell(a_T - \epsilon) - \ell(a_T)}{-\epsilon}
                - \beta \frac{\ell(\epsilon) - \ell(0)}{\epsilon}
            \right\}.
    \end{equation*}
    The term inside the parentheses converges to 
    \begin{equation*}
        \ell'(a_T) - \beta \ell'(0) > \ell'(\eta) - \beta \ell'(0) = 0,
    \end{equation*}
    where the final equality is due to $\eta < \hat x$ and Lemma~\ref{l:eta}.
    Once again we conclude that for $\epsilon$ sufficiently small, the
    difference $\sum_{t=0}^{\infty} \beta^t \ell(a_t) - \sum_{t=0}^{\infty}
    \beta^t\ell(\hat a_t)$ is positive, contradicting optimality.
\end{proof}

\begin{proof}[Proof of Proposition~\ref{p:policy}]
    To show the equivalence between \eqref{eq:dpm} and \eqref{eq:epe}, we
    first show that \eqref{eq:dpm} is equivalent to $\bar{w} = \inf_{\mu \in
      \mM} w_\mu$ where $w_\mu$ is as defined in \eqref{eq:piv}. Suppose
    that the optimal policy is
    $\mu = (\pi_0, \pi_1, \ldots)$ and we let $\sigma_t(x) = x - \pi_t(x)$.
    Then we have
    \begin{multline}
        \label{eq:fmu}
        \bar{w}(\hat x) = w_\mu(\hat x) = \ell[\pi_0(\hat x)] + \beta
        \ell[\pi_1\sigma_0(\hat x)] +
        \beta^2 \ell[\pi_2\sigma_1\sigma_0(\hat x)] + \ldots \\
        + \limsup_{t\to \infty} \beta^k
        \ell'(0)\sigma_{t-1}\sigma_{t-2}\cdots \sigma_0(\hat x).
    \end{multline}
    It is clear that $\bar{w}$ is finite. Therefore, the optimal policy must
    satisfy $\sigma_t \to 0$, otherwise the last term in \eqref{eq:fmu}
    would go to infinity. Let $a_t = \pi_t\sigma_{t-1}\ldots
    \sigma_0(\hat x)$. We claim that $\{a_t\}$ solves \eqref{eq:dpm}.
    Suppose not and the solution to (\ref{eq:dpm}) is $\{a'_t\}$. Then by
    Lemma~\ref{l:finn0}, $a'_t = 0$ for all $t > T$ for some $T$. Thus we
    can construct a policy $\mu'$ that reproduces $\{a_t'\}$ and gives a
    lower loss. This is a contradiction. Conversely, suppose that the
    solution to \eqref{eq:dpm} is $\{a_t\}$. Using the same argument, we can
    show that the policy that gives rise to $\{a_t\}$ is an optimal policy.
    Therefore, $W = \bar{w}$.

    Next we show that $w^* = \bar{w}$ using Theorem~\ref{t:opt}. Both
    conditions in Assumption~\ref{a:contbeta} can be verified for $(L, G)$.
    Part 1 of Assumption~\ref{a:contbeta} is trivial in this setting, since
    $v_n \uparrow v$ pointwise clearly implies $\ell(a) + \beta v_n(x-a) \to
    \ell(a) + \beta v(x-a)$ at each $(x, a) \in F_G$. Part~2 also holds,
    since for any $r > 0$ and $w \geq \phi$, we have
    \begin{equation*}
      L(x, a, w + r) = \ell(a) + \beta w(x-a) + \beta r = L(x, a, w) + \beta
      r.
    \end{equation*}
    Hence Theorem~\ref{t:opt} applies. It follows from Theorem~\ref{t:opt}
    that $w^* = \bar{w}$, there exists an stationary optimal policy, and the
    Bellman's principle of optimality holds. Since $\pi^*$ satisfies
    $T_{\pi^*} w^* = T w^*$, $\pi^*$ is a stationary optimal policy.
    
    Theorems~\ref{t:bk0} and \ref{t:con} imply that $\pi^*$ is continuous
    and single-valued. It then follows from the principle of optimality that
    $\{a_t^*\}$ is the unique solution to \eqref{eq:dpm}.
\end{proof}

\begin{proposition} \label{p:finite}
    For all $n\in\NN$ and increasing convex $w\in\iI$, we have
    \begin{equation*}
        T^n w(x) = w^*(x) \; \text{ whenever } \, x \leq n \eta.
    \end{equation*}
\end{proposition}

Proposition~\ref{p:finite} implies uniform convergence in \emph{finite}
time. In particular, for $n \geq \hat x / \eta$ we have $T^n w = w^*$
everywhere on $[0, \hat x]$. Note that this bound $\hat x / \eta$ is
independent of the initial condition $w$.

\begin{proof}[Proof of Proposition~\ref{p:finite}]
    It suffices to show that if $f, g \in \fF$, then $T^k f = T^k g$ on $[0,
    k \eta]$. We prove this by induction.

    To see that $T^1 f = T^1 g$ on $[0, \eta]$, pick any $x \in [0, \eta]$
    and recall from Lemma~\ref{l:comp2} that if $h \in \fF$ and $x \leq
    \eta$, then $Th(x) = \ell(x)$. Applying this result to both $f$ and $g$
    gives $Tf(x) = Tg(x) = \ell(x)$. Hence $T^1 f = T^1 g$ on $[0, \eta]$ as
    claimed.

    Turning to the induction step, suppose now that $T^k f = T^k g$ on $[0,
    k \eta]$, and pick any $x \in [0, (k+1)\eta]$. Let $h \in \fF$ be
    arbitrary, let $\pi_h$ be the $h$-greedy function, and let $\sigma_h(x)
    := x - \pi_h(x)$. By Lemma~\ref{l:tstarnz}, we have $\pi_h(x) \geq
    \eta$, and hence
    \begin{equation*}
        \sigma_h(x) \leq x - \eta \leq (k+1) \eta - \eta \leq k \eta.
    \end{equation*}
    In other words, given function $h$, the optimal choice at $x$ is less
    than $k \eta$. Since this is true for both $h = T^k f$ and $h = T^k g$,
    we have
    \begin{equation*}
        T^{k+1}f(x) 
        = \min_{0 \leq y \leq x} \{ \ell(x - y) + \beta T^k f(y) \} 
        = \min_{0 \leq y \leq k \eta} \{ \ell(x - y) + \beta T^k f(y) \}.
    \end{equation*}
    Using the induction step we can now write
    \begin{equation*}
        T^{k+1}f(x) 
         = \min_{0 \leq y \leq k \eta} \{ \ell(x - y) + \beta T^k g(y) \} 
         = \min_{0 \leq y \leq x} \{ \ell(x - y) + \beta T^k g(y) \} .
    \end{equation*}
    The last expression is just $T^{k+1} g(x)$, and we have now shown that
    $T^{k+1} f = T^{k+1} g$ on $[0, (k+1) \eta]$.  The proof is complete.
\end{proof}

\begin{proof}[Proof of Proposition~\ref{p:co22}]
    Since $\ell'(0) = 0$, \eqref{eq:EU} is equivalent to
    $\beta\ell'(a^*_{t+1}) = \ell'(a^*_t)$.

  \textbf{Sufficiency}. Let $x_0^* = \hat x$ and $x_t^* = x_{t-1}^* -
  a_{t-1}^*$ for $t\geq 1$. Let $\{a_t\}$ be any feasible sequence. Let $x_0
  = \hat x$ and $x_t = x_{t-1} - a_{t-1}$. It suffices to prove that
  \begin{equation*}
    D := \lim_{T \to\infty} \sum_{t=0}^T \beta^t [\ell(a_t^*) - \ell(a_t)]
    \leq 0.
  \end{equation*}
  Since $\ell$ is convex, we have
  \begin{equation*}
    D 
    = \lim_{T \to\infty} \sum_{t=0}^T \beta^t [\ell(x_t^* - x_{t+1}^*) -
    \ell(x_t - x_{t+1})]
    \leq \lim_{T \to\infty} \sum_{t=0}^T \beta^t \ell'(a_t^*) (x_t^* - x_t
    - x_{t+1}^* + x_{t+1}).
  \end{equation*}
  Since $x_0 = x_0^*$, rearranging gives
  \begin{equation*}
    D \leq \lim_{T \to\infty} \sum_{t=0}^T \beta^t (x_{t+1}^* -
    x_{t+1})[\beta \ell'(a_{t+1}^*) - \ell'(a_t^*)] - \beta^T
    \ell'(a_T^*)(x_{T+1}^* - x_{T+1}).
  \end{equation*}
  Since $\beta \ell'(a_{t+1}^*) = \ell'(a_t^*)$, the summation is zero and
  $\beta^T \ell'(a_T^*) = \ell'(a_0^*)$. We have
  \begin{equation*}
    D \leq -\lim_{T \to\infty} \ell'(a_0^*) (x_{T+1}^* - x_{T+1}).
  \end{equation*}
  Since $\{a_t\}$ and $\{a_t^*\}$ are feasible, $x_{T+1}$ and $x_{T+1}^*$ go
  to zero as $T\to\infty$. Hence $D \leq 0$.

  \textbf{Existence and Uniqueness.} Since $\{a_t^*\}$ is feasible and
  satisfies $\beta \ell'(a_{t+1}^*) = \ell'(a_t^*)$ for all $t$, we have
  \begin{equation*}
    \hat x = \sum_{t=0}^\infty a_t^* = \sum_{t=0}^\infty
    (\ell')^{-1}\left(\frac{1}{\beta^t} \ell'(a_0^*)\right) =: g(a_0^*),
  \end{equation*}
  where $(\ell')^{-1}$ is well defined on $[0, \lim_{x\to \infty} \ell'(x)]$
  because $\ell$ is increasing, strictly convex, and $\ell'(0) = 0$. Hence,
  $g$ is well defined on $\RR_+$ and $g(a_0^*)$ is continuous and strictly
  increasing in $a_0^*$. Since $g(0) = 0$ and $g(\hat x) > \hat x$, there exists a
  unique $a_0^*>0$ such that $\{a_t^*\}$ satisfying $\beta \ell'(a_{t+1}^*)
  = \ell'(a_t^*)$ is feasible, $a_t^* > 0$ for all $t$, and $\{a_t^*\}$ is
  strictly decreasing. That $\{a_t^*\}$ is an optimal solution then follows
  from the sufficiency part. Since $\ell$ is strictly convex, the solution
  is unique.

  \textbf{Necessity.} Since we have pinned down a unique solution of
  \eqref{eq:dpm} which satisfies $\beta \ell'(a_{t+1}^*) = \ell'(a_t^*)$,
  the condition is also necessary.
\end{proof}

\subsection{Proofs for Section~\ref{s:app:production}}

\begin{proof}[Proof of Proposition~\ref{p:oie}]
    We must verify that $(W, \{a_i^*\})$ satisfies Definition~\ref{d:eq}.
    We first consider the case of $\ell'(0) > 0$. By
    Propositions~\ref{p:ndp} and \ref{p:policy}, the value function $W$ is a
    solution to the Bellman equation \eqref{eq:epe}, and hence satisfies
    \begin{equation}
        \label{eq:sb}
        W(s) = \min_{0 \leq v \leq s} \, \{ c(v) + (1 + \tau) W(s-v) \}
        \quad \text{for all } s \in [0, 1],
    \end{equation}
    and $W$ lies in the class $\fF$ of increasing,
    convex and continuous functions $f \colon \RR_+ \to \RR_+$ such that
    $c'(0)s \leq f(s) \leq c(s)$ for all $s \in \RR_+$. In addition, with
    $\{x_i\}$ as the optimal state process (see Proposition~\ref{p:policy}),
    we have,
    \begin{equation}
        \label{eq:sbp}
        W(x_i) = \{ c(a^*_i) + (1 + \tau) W(x_{i+1}) \}
        \quad \text{for all } i \geq 0.
    \end{equation}
    We need to show that 1--3 of Definition~\ref{d:eq} hold when $p=W$ and
    $v_i = a_i^*$ for all $i\geq 0$. Part~1 is immediate because $W \in \fF$ and
    all functions in $\fF$ must have this property, while Part~2 follows
    directly from \eqref{eq:sb}. To see that Part~3 of Definition~\ref{d:eq}
    also holds, let $b_i = x_i$. By the definition of the state process, the
    sequence $\{b_i\}$ then corresponds to the downstream boundaries of a
    set of firms obeying task allocation $\{a_i^*\}$. The profits of firm
    $i$ are $\pi_i = W(b_i) - c(a^*_i) - (1 + \tau) W(b_{i+1})$. By
    \eqref{eq:sbp} and $b_i = x_i$, we have $\pi_i = 0$ for all $i$. Hence
    Part~3 of Definition~\ref{d:eq} also holds, as was to be shown.

    If $\ell'(0) = 0$, part 1 follows from the definition of the value
    function \eqref{eq:vf}. By Proposition~\ref{p:co22}, for any $t$ with $0\leq
    t \leq 1$, there exists a unique optimal allocation $\{a_{t, j}^*\}$
    such that $W(t) = \sum_j \beta^j \ell(a_{t, j}^*)$, and $\sum_j a_{t,
      j}^* = t$. Since $\{s-t, a_{t, 0}^*, a_{t, 1}^*, \ldots\}$ is a
    feasible allocation at stage $s$ with $t\leq s \leq 1$, part 2 follows
    from the definition of the value function. To see part 3, let $b_0 = 1$
    and $b_i = b_{i-1} - a_{i-1}^*$. By Proposition~\ref{p:co22}, we have
    $\ell'(a_i^*) = (1+\tau)\ell'(a_{i+1}^*)$. Since $\sum_{i=j}^\infty
    a_i^* = b_j$ for all $j$, it follows again from Proposition~\ref{p:co22}
    that $\{a_i^*\}_{i=j}^\infty$ is an optimal allocation for stage $b_j$.
    Therefore, $p(b_i) = \sum_{j=0}^\infty (1+\tau)^j c(a_{i+j}^*) =
    c(a_i^*) + (1+\tau)p(b_{i+1})$ for all $i$. Hence, $\pi_i = 0$ for all
    $i$.
\end{proof}

\subsection{Proofs for Section~\ref{s:app:networks}}\label{ss:networks-proofs}

\begin{proof}[Proof of Proposition~\ref{p:eq_mp}]
    To study this problem in the framework of Section~\ref{ss:gdp}, we set $X
    = [0, \hat x]$, $A = [0, \hat x] \times \NN$, $G(x) = [0, x] \times
    \NN$, and
    \begin{equation*}
        L(x, a, w) = c(x-t) + g(k) + (1+\tau)kp(t/k) \qquad a = (t, k).
    \end{equation*}
    Since $g(k) \to \infty$ as $k \to \infty$, we can restrict $G(x)$ to be
    $[0, x] \times \{1, 2, \ldots, \bar{k}\}$ so that $G$ is compact-valued.
    Under the conditions of Proposition~\ref{p:eq_mp}, it can be shown that
    $A_1$--$A_5$ hold with $\psi = c$ and $\phi(s) = c'(0)s$ (see
    \cite{yu2019equilibrium}). Then, Theorem~\ref{t:bk0} implies that the
    Bellman equation~\eqref{eq:mp} has a unique solution $p^*$ in $\iI$,
    $T^np \to p^*$ for all $p\in\iI$ where
	\begin{equation*}
	(Tp)(s) := \min_{\substack{0 \leq t\leq s\\ k\in\NN}} L(x, a, w),
	\end{equation*}
    and $t^*$ and $k^*$ exist. We need only verify that $(p^*, \{v_i\}, \{k_i\})$
    given by $v_i = b_i - t^*(b_i)$, $k_i = k^*(b_i)$ and $b_{i+1} = (b_i -
    v_i)/k_i$ is an equilibrium, the definition of which is given in
    Section~\ref{ss:network}.

	Since $p^*\in\iI$, $p(0) = 0$. Since $p^*$ satisfies \eqref{eq:mp}, part
    (ii) of the definition is also satisfied. To see that part (iii) holds,
    note that
	\begin{align*}
	p^*(b_i) &= c(b_i - t^*(b_i)) + g(k^*(b_i)) +
	(1+\tau) k^*(b_i) p^*\left(\frac{t^*(b_i)}{k^*(b_i)}\right)\\
	&= c(v_i) + g(k_i) + (1+\tau) k_i p^*
	\left(
	\frac{b_i - v_i}{k_i}
	\right).
	\end{align*}
	It follows that $\pi_i = 0$ for all $i\in\ZZ$ where $\pi_i$ is as
	defined in \eqref{eq:prof_mp}. This completes the proof.
\end{proof}

\bibliographystyle{ecta}

\bibliography{jet_bib}

\end{document}